\documentclass[aps,prx,10pt]{revtex4-2}

\usepackage[margin=1.1in]{geometry}
\usepackage{float}
\usepackage{graphicx}
\usepackage{dcolumn}
\usepackage{bm}

\usepackage{wustyle}

\usepackage{braket}
\usepackage[english]{babel}
\usepackage{amsmath, amssymb, amsthm}
\usepackage{mathtools}
\usepackage{tikz-cd} 
\usepackage{algorithm}
\usepackage{algpseudocode}
\usepackage{amsmath}
\usepackage{amsfonts}

\newcommand{\id}{\mathrm{id}}

\newcommand{\cO}[0]{\mathcal{O}}

\newcommand{\q}{\quad}
\newcommand{\mc}{\mathcal}
\newcommand{\maf}{\mathfrak}

\newcommand{\rd}{\mathrm{d}}

\usepackage{enumitem}

\usepackage{amsmath, amssymb, amsthm}
\usepackage{physics} 
\usepackage{mathtools}

\begin{document}

\title{Thermal expectation estimation via single-trajectory Gibbs sampling with non-destructive measurements}

\author{Hongrui Chen}
\email{hongrui@stanford.edu}
\affiliation{Department of Mathematics, Stanford University, Stanford, CA 94305, USA}

\author{Jiaqing Jiang}
\email{jiaqingjiang95@gmail.com}
\affiliation{Simons Institute for the Theory of Computing, University of California, Berkeley, Berkeley, CA 94720, USA}

\author{Bowen Li}
\email{bowen.li@cityu.edu.hk}
\affiliation{Department of Mathematics, City University of Hong Kong, Kowloon Tong, Hong Kong SAR}

\author{Lexing Ying}
\email{lexing@stanford.edu}
\affiliation{Department of Mathematics, Stanford University, Stanford, CA 94305, USA}
\affiliation{Institute for Computational and Mathematical Engineering, Stanford University, Stanford, CA 94305, USA}

\date{\today}

\begin{abstract}
Estimating thermal expectation values of quantum many-body systems is a central challenge in physics, chemistry, and materials science. Standard quantum Gibbs sampling protocols address this task by preparing the Gibbs state from scratch after every measurement, incurring a full mixing-time cost at each step. 
Recent advances in single-trajectory Gibbs sampling \cite{jiang2026} substantially reduce this overhead: once stationarity is reached, measurements can be collected along a single trajectory without re-thermalizing, provided the measurement channel preserves the Gibbs ensemble. However, explicit constructions of such non-destructive measurements have been limited primarily to observables that commute with the Hamiltonian. 
In this work, we fundamentally extend the single-trajectory framework to arbitrary, non-commuting observables. We provide two measurement constructions that extract measurement information without fully destroying the Gibbs state, thereby eliminating the need for full re-mixing between samples.  

First, we construct a measurement that satisfies exact detailed balance. This ensures the system remains in equilibrium throughout the trajectory, allowing measurement outcomes to decorrelate in an autocorrelation time that could be significantly shorter than the global mixing time. Second, assuming the underlying quantum Gibbs sampler has a positive spectral gap, we design a simplified measurement scheme that ensures the post-selected state serves as a warm start for rapid re-mixing. This approach successfully decouples the resampling cost from the global mixing time. Both measurement schemes admit efficient quantum circuit implementations, requiring only polylogarithmic Hamiltonian simulation time.
\end{abstract} 

\maketitle

\section{Introduction}

Estimating thermal expectation values is a central task in quantum physics and 
quantum chemistry, underpinning our theoretical understanding of quantum 
many-body systems at finite temperatures. At inverse temperature $\beta$, a 
quantum system in thermal equilibrium is governed by the Gibbs state 
$\sigma_\beta = e^{-\beta H}/\mathcal{Z}$, where $H$ is the Hamiltonian and 
$\mathcal{Z} = \mathrm{Tr}(e^{-\beta H})$ is the partition function. Key properties of the system, such as energy, magnetization, and correlation functions, are encoded in the thermal expectation value, defined as $\langle A \rangle_{\sigma_{\beta}}= \mathrm{Tr}(\sigma_{\beta} A)$ for the respective observable $A$. Computing thermal expectation values for general quantum systems is intractable on classical computers due to the sign problem and the exponential growth of the Hilbert space, which motivates the development of quantum algorithms. 

Recent advances in quantum Gibbs sampling~\cite{Chen2025efficient, chen2025efficientexactnoncommutativequantum, Ding_2025, jiang2024weak,temme2011quantum} provide a foundation for estimating thermal expectation values on quantum computers, opening the possibility of investigating quantum many-body systems beyond classical computational limits. The quantum Gibbs sampling algorithm, which serves as a quantum analogue of classical Monte Carlo methods, can drive any initial state to the Gibbs state within a timescale known as the \emph{mixing time} $t_{\mathrm{mix}}$. To estimate thermal expectation values to precision $\epsilon$, it suffices to prepare $\mathcal{O}(1/\epsilon^2)$ Gibbs states and perform measurements, resulting in a total Gibbs sampling time of $\mathcal{O}(t_{\mathrm{mix}}/\epsilon^2)$. While this approach provides a viable method for estimating thermal expectation values, its $\mathcal{O}(t_{\mathrm{mix}}/\epsilon^2)$ scaling can be prohibitively large, motivating the development of techniques that further 
reduce the sampling cost.

To address these challenges, \cite{jiang2026} introduced a new paradigm for thermal expectation estimation, demonstrating that the sampling cost can be substantially reduced using a single Gibbs-sampling trajectory. Their algorithm first runs the Gibbs sampler for roughly one mixing time to bring the given initial state close to the Gibbs state, then performs coherent measurements at regular intervals 
to collect samples. The key observation is that if the measurement channel $\mathcal{M}$ satisfies detailed balance and hence preserves the Gibbs 
state, i.e., $\mathcal{M}(\sigma_\beta) = \sigma_\beta$, then the chain remains in the Gibbs ensemble during the sampling stage. Consequently, once the chain has reached stationarity, consecutive samples become effectively independent on a timescale much shorter than the mixing time. This timescale is captured by the \emph{autocorrelation time} $t_{\rm aut}$, which quantifies the rate of decorrelation in stationarity. 
As a result, to estimate thermal expectation values to accuracy $\epsilon$, the total Gibbs sampling cost is $t_{\mathrm{mix}} + \mathcal{O}(t_{\mathrm{aut}}/\epsilon^2)$, which is substantially more efficient than  $\mathcal{O}(t_{\mathrm{mix}}/\epsilon^2)$ whenever $t_{\mathrm{aut}} \ll t_{\mathrm{mix}}$. \cite[Section 1.3]{jiang2026} also provided examples in which  $t_{\mathrm{aut}}$ is significantly smaller than $ t_{\mathrm{mix}}$. More generally, they showed that the autocorrelation time $t_{\mathrm{aut}}$ is upper bounded by the inverse spectral gap of the quantum Gibbs sampler, which in principle could be smaller than the mixing time by a factor proportional to the system size.

As mentioned, the key ingredient enabling the sampling reduction in \cite{jiang2026} is the non-destructive measurement of Gibbs states, meaning that it preserves the Gibbs state ensemble ($\cM(\sigma_{\beta})=\sigma_{\beta}$).
For observables that commute with the Hamiltonian, they explicitly construct a non-destructive measurement with only logarithmic overhead via Gaussian-filtered quantum phase estimation~\cite{moussa2019low}. 
For general observables, under additional assumptions, they showed that the weighted operator Fourier transform~\cite{chen2025efficientexactnoncommutativequantum} can be used to mitigate measurement-induced disturbance. However, the construction of non-destructive measurements for general observables remains an open problem.

In this work, we close this gap by constructing non-destructive measurement channels (satisfying detailed balance) for general observables that may not commute with the Hamiltonian. As a consequence, our construction fully extends the single-trajectory framework to the non-commuting setting.

 \begin{theorem}[Exact detailed-balance measurement (informal version of Theorem~\ref{thm:DB-channel})]\label{thm:informal}
Let $\sigma_{\beta} \propto e^{-\beta H}$ be the Gibbs state of a Hamiltonian $H$ at inverse temperature $\beta$. Consider any Hermitian observable $A$ with $\|A\|\le 1$. One can construct a  quantum channel $\mathcal{M}$,
\begin{itemize}
\item $\mathcal{M}$ satisfies KMS detailed balance with respect to $\sigma_{\beta}$, and hence fixes the Gibbs state $\mathcal{M}(\sigma_{\beta})=\sigma_{\beta}$.
\item The outcomes of $\mathcal{M}$ yield an unbiased estimator of $\mathrm{Tr}(\sigma_{\beta} A)$ of $\widetilde{\mathcal{O}}(1)$ variance.
\item For any implementation accuracy $\delta$, a $\delta$-approximate implementation of $\mathcal{M}$ can be realized with
\begin{align*}
&\widetilde{\mathcal{O}}(1) \text{ queries to a block encoding of } A,\\
&\widetilde{\mathcal{O}}(\beta) \text{ controlled Hamiltonian simulation time for } H, \text{ and}\\
&\widetilde{\mathcal{O}}(1) \text{ ancilla qubits.}
\end{align*}
\end{itemize}
where $\widetilde{\mathcal{O}}$ suppresses polylogarithmic factors in $\beta$, $\|H\|$, $1/\delta$, and $1/\epsilon$.
Combining the constructed measurement with the single-trajectory Gibbs sampling approach~\cite{jiang2026}, one can estimate $\mathrm{Tr}(\sigma_{\beta} A)$ to precision $\epsilon$ with failure probability $\eta$ using a total Gibbs sampling time of $t_{\mathrm{mix}}+\widetilde{\mathcal{O}}(t_{\mathrm{aut}}/(\epsilon^2\eta))$, where $t_{\mathrm{aut}}$ denotes the corresponding autocorrelation time and is bounded by the inverse spectral gap (up to polylogarithmic factor) of the Gibbs sampling channel.
\end{theorem}

Theorem~\ref{thm:informal} implies that, for general observables, our constructed measurement combined with the single-trajectory Gibbs sampling approach~\cite{jiang2026} allows one to obtain effectively independent samples every autocorrelation time rather than per mixing time, leading to a significant reduction in the Gibbs sampling cost for thermal expectation estimation.
As illustrated in the examples of \cite{jiang2026}, in many scenarios the autocorrelation time $t_{\mathrm{aut}}$ can be much shorter than the mixing time (potentially by a sub-exponential factor), since it leverages a warm start from the previous measurement and depends on the observable of interest. More precisely, $t_{\mathrm{aut}}$ is upper bounded by the inverse spectral gap 
$\lambda^{-1}$, whereas the mixing time may carry the additional factor 
$\log\,(\sigma_{\beta,\mathrm{min}}^{-1})$, where $\sigma_{\beta,\mathrm{min}}$ 
denotes the smallest eigenvalue of $\sigma_{\beta}$ that could be exponentially small in the number of qubits $n$. The single-trajectory approach avoids this prefactor entirely in the per-sample cost.

In addition to Theorem \ref{thm:informal}, we also construct a simpler measurement that forgoes the detailed-balance property but instead guarantees that the post-selected state remains a warm start for the Gibbs sampler. Under a spectral gap assumption, this ensures that each remixing step costs only inverse-spectral-gap time rather than the full mixing time.

\begin{theorem}[Informal version of Theorem \ref{thm:complexity1}]\label{thm:informal_ws}
Let $\sigma_{\beta}$, $H$, and $A$ be as in Theorem~\ref{thm:informal}. Suppose that $\mathcal{N}$ is a primitive quantum Gibbs sampling channel with unique fixed point $\sigma_\beta$, satisfying KMS detailed balance, with spectrum contained in $[0,1]$, and spectral gap $\lambda > 0$. One can construct a measurement channel $\mathcal{M}$ such that:
\begin{itemize}
\item If the input state is a $\chi^2$-warm start, then each post-selected state remains a $\chi^2$-warm start; reaching the Gibbs state again within $\chi^2$-accuracy $\epsilon$ requires only $\mathcal{O}(\lambda^{-1}\log(1/\epsilon))$ Gibbs sampling time.
\item At stationarity, the outcomes of $\mathcal{M}$ yield an unbiased estimator of $\mathrm{Tr}(\sigma_{\beta} A)$ of $\mathcal{O}(1)$ variance.
\end{itemize}
Consequently, one can estimate $\mathrm{Tr}(\sigma_{\beta} A)$ to precision $\epsilon$ with failure probability $\eta$ using a total Gibbs sampling time of $t_{\mathrm{mix}} + \widetilde{\mathcal{O}}(\lambda^{-1}\log(1/\eta)/\epsilon^2)$.
\end{theorem}

In the following section, we give an overview of the construction of our two measurement channels: an exact detailed-balance channel, and a simplified channel that does not preserve detailed balance but guarantees a constant $\chi^2$-divergence warm start. Readers interested in further details on the single-trajectory approach may refer to Section \ref{sec:single-traj-DB} or \cite{jiang2026}.

\subsection{Overview of our Construction}\label{sec:construction}

 Recall that we use $\sigma_{\beta} \propto \exp(-\beta H)$ to denote the Gibbs state with respect to Hamiltonian $H$ and inverse temperature $\beta$. Consider an observable $A$ with $\|A\|\leq 1$. Our goal is to design measurement channels that simultaneously extract an unbiased estimator of the thermal expectation value $\mathrm{Tr}(\sigma_{\beta} A)$ and leave the post-measurement state in a 
controlled proximity to $\sigma_\beta$, either by exactly preserving the Gibbs 
state via detailed balance, or by guaranteeing a bounded $\chi^2$-divergence warm 
start that enables fast remixing.

To build some intuition, note that for observables commuting with $H$, i.e., 
$[A, H] = 0$, one can measure $A$ without disturbing the Gibbs state ensemble, 
for instance via projective measurement in the energy eigenbasis or quantum phase 
estimation. For general non-commuting observables, however, projective measurement 
can significantly perturb the Gibbs state, motivating the need for more carefully designed measurement channels. 

A natural approach to mitigating measurement disturbance, also discussed in 
\cite{jiang2026, chen2025catalytic}, is to smooth the observable $A$ via the 
operator Fourier transform,
\begin{align}
    A_f = \int_{-\infty}^{+\infty} f(t)\, e^{iHt} A e^{-iHt}\, dt, \label{eq:Af}
\end{align}
where $f(t) = \frac{1}{\sqrt{2\pi\tau^2}}\exp(-t^2/2\tau^2)$ is a Gaussian with 
width parameter $\tau$. The smoothed observable $A_f$ has several desirable 
properties \cite{jiang2026}. First, it preserves the thermal expectation value, 
$\mathrm{Tr}(A_f \sigma_{\beta}) = \mathrm{Tr}(A \sigma_{\beta})$, so one can 
equivalently measure $A_f$ in place of $A$. Second, when $\tau$ is large, $A_f$ 
approximately commutes with $\sigma_{\beta}$, so its measurement incurs only a 
slight disturbance to the Gibbs state.
However, the implementation cost of the measurement of $A_f$ scales
as $\widetilde{\mathcal{O}}(\tau)$,
while ensuring approximate commutation with $H$ in the
non-commuting case generally requires $\tau$ to be large,
leading to significant
overhead. Besides, the resulting measurement only approximately preserves the Gibbs 
state and does not satisfy exact detailed balance, so the single-trajectory sample 
complexity theorem of \cite{jiang2026} does not directly apply.

We present two constructions that address these limitations in different ways. In the first construction, we show that it is possible to construct a measurement that satisfies exact detailed balance, based on $A_f$ with a moderate $\tau = \tilde{\cO}(\beta)$.  In the second construction, we develop a measurement that perturbs the Gibbs state but guarantees a constant $\chi^2$-divergence warm start, enabling fast remixing back to the Gibbs state. 

We begin with a brief review of relevant concepts; see Section \ref{sec:preli} for details. A quantum measurement channel is 
described by its Kraus operators $\{K_a\}_a$, satisfying the trace-preserving 
condition $\sum_a K_a^\dagger K_a = I$, under which a state $\rho$ is mapped to 
$\sum_a K_a \rho K_a^\dagger$. A quantum channel (or more generally, a completely 
positive map) is said to satisfy the Kubo--Martin--Schwinger (KMS) detailed balance 
condition with respect to $\sigma_\beta$ if it is self-adjoint under the KMS inner 
product, or equivalently, if it satisfies
\begin{align*}
   \sum_a  \sigma_{\beta}^{-1/2} K_a \rho K_a^\dagger \sigma_{\beta}^{-1/2}
    = \sum_a K_a^\dagger \sigma_{\beta}^{-1/2} \rho\, \sigma_{\beta}^{-1/2} K_a.
\end{align*}
Setting $\rho = \sigma_\beta$ and using the 
trace-preserving property, one verifies that any KMS detailed-balanced channel 
necessarily fixes the Gibbs state: $\mathcal{M}(\sigma_\beta) = \sigma_\beta$.

\medskip \noindent
\textbf{Exact detailed-balance measurement for non-commuting observables.} 
Our first construction is a measurement channel with Kraus operators
\begin{align} \label{def:kraus}
    K_1 = c\sqrt{I + uA_f}, \qquad K_2 = c\sqrt{I + uA_{\tilde{f}}}, \qquad
    K = \sqrt{\sqrt{\sigma_{\beta}}(I - K_1^\dagger K_1 - K_2^\dagger K_2)
\sqrt{\sigma_{\beta}}}\,\sigma_{\beta}^{-1/2}.
\end{align}
We now explain the role of each component and why the construction satisfies KMS detailed balance; its estimation properties and implementation cost are discussed subsequently. A more technical treatment is provided in Section \ref{sec:DBmeasure}. 

Here $f(t) = \frac{1}{\sqrt{2\pi\tau^2}}\exp(-t^2/2\tau^2)$ is the Gaussian filter 
with width $\tau$, and $\tilde{f}(t) := f(t - i\beta/2)$ is its imaginary-time shift 
by $i\beta/2$. The shifted filter is chosen precisely to enforce the algebraic relation $K_2 = \sigma_{\beta}^{1/2} K_1^\dagger \sigma_{\beta}^{-1/2}$, 
which, as shown in Lemma~\ref{lem:DBKraus}, is sufficient to guarantee that the 
completely positive (CP) map $\mathcal{T}'[\rho] = K_1 \rho K_1^\dagger + K_2 \rho K_2^\dagger$ 
satisfies KMS detailed balance with respect to $\sigma_\beta$. Since $\mathcal{T}'$ is not trace-preserving, it does not by itself define a valid quantum channel. Following the discrete-time quantum Metropolis--Hastings framework of \cite{gilyen2026}, we append a rejection branch with Kraus operator $K$ that restores the trace-preserving property while maintaining the exact KMS detailed balance.  The parameters $c$ and $u$ are rescaling factors to ensure the square root function in the definition \eqref{def:kraus} of Kraus operators is well-defined. To avoid confusion, we note that since $\tilde{f}$ is not a real symmetric function, the operator $A_{\tilde{f}}$ is not Hermitian. The square root in $K_2$ is defined via an absolutely convergent Taylor expansion.

The thermal expectation value is naturally extracted from the measurement outcomes 
as follows. Assign outcome values $z = 1$ to branches $K_1$ and $K_2$, and $z = 0$ 
to the rejection branch $K$. When the system is in the Gibbs state $\sigma_\beta$, 
the expected outcome $\mathbb{E}[z]$ equals the total probability of the non-rejected 
branches:
\begin{align}
    \mathrm{Tr}(K_1 \sigma_{\beta} K_1^\dagger) + \mathrm{Tr}(K_2 \sigma_{\beta} K_2^\dagger)
    = 2c^2(1 + u\,\mathrm{Tr}(A_f \sigma_{\beta}))
    = 2c^2(1 + u\,\mathrm{Tr}(A \sigma_{\beta})),
\end{align}
where the first equality uses $K_2 = \sigma_\beta^{1/2} K_1^\dagger \sigma_\beta^{-1/2}$ 
and the second uses the fact that $A_f$ preserves the thermal expectation value, 
$\mathrm{Tr}(A_f \sigma_\beta) = \mathrm{Tr}(A\sigma_\beta)$. It follows that the 
random variable
\begin{align}
    v := \frac{z}{2c^2 u} - \frac{1}{u}
\end{align}
is an unbiased estimator of $\mathrm{Tr}(A\sigma_\beta)$, with variance of order 
$(cu)^{-2}$.

The performance of the constructed measurement is governed by the parameters $c$ 
and $u$. A crucial property of the smoothed operators is that the norms $\|A_f\|$ 
and $\|A_{\tilde{f}}\|$ remain $\mathcal{O}(1)$ when $\tau = \Omega(\beta)$. 
Consequently, it suffices to set $\tau = \Theta(\beta)$ and $c, u = 
\widetilde{\mathcal{O}}(1)$ to ensure that $u\|A_f\|$ and $u\|A_{\tilde{f}}\|$ 
are bounded away from $1$, so that the matrix square roots in the Kraus operators 
are well-defined with rapidly convergent Taylor expansions. Based on this choice 
of parameters and the techniques for implementing the rejection operator $K$ from 
\cite{gilyen2026}, we implement the measurement to precision $\epsilon$ using 
$\widetilde{\mathcal{O}}(1)$ queries to a block encoding of $A$, 
$\widetilde{\mathcal{O}}(\beta)$ controlled Hamiltonian simulation time, and 
$\widetilde{\mathcal{O}}(1)$ ancilla qubits, via the linear combination of 
unitaries and quantum singular value transformation~\cite{Gily_n_2019}.

\medskip \noindent
\textbf{Measurement-remix strategy via warm starts in $\chi^2$-divergence.}
Our second construction is simpler. We define a two-outcome measurement with Kraus 
operators
\begin{align}
    K_\pm = \sqrt{\frac{1}{2}(I \pm uA_f)},
\end{align}
where the outcome $\pm$ is assigned the value $z = \pm 1$. Upon observing outcome 
$\pm$, the post-selected state is $\rho'_\pm = K_\pm \rho K_\pm^\dagger / p_\pm$, 
and the difference in outcome probabilities $p_+ - p_- = u\,\mathrm{Tr}(\rho A_f)$ 
yields an unbiased estimator of $\mathrm{Tr}(\sigma_\beta A)$. Since this 
construction uses only the real-time smoothed observable $A_f$, requiring no 
imaginary-time shift and no rejection branch, it is considerably simpler to 
implement than the above exact detailed-balance channel. 

However, this measurement channel does not satisfy exact detailed balance, and the 
post-measurement state is disturbed. The key result is that the disturbance remains 
controlled: with $\tau = \Omega(\beta)$, the post-selected state $\rho'_\pm$ 
satisfies $\chi^2(\rho'_\pm, \sigma_\beta) = \mathcal{O}(1)$, providing a warm 
start for the subsequent Gibbs sampling dynamics. Under a spectral gap $\lambda > 0$, this guarantees that the state $\rho'_\pm$ returns to $\epsilon$-proximity of 
$\sigma_\beta$ in remixing time $\mathcal{O}(\lambda^{-1}\log(1/\epsilon))$, 
eliminating the $\mathcal{O}(n)$ overhead from $\log\,(\sigma_{\beta,\mathrm{min}}^{-1})$ 
that appears in the general mixing time estimate. We refer to Section~\ref{sec:warm-start} 
for the precise statement and proof.

\subsection{Related Works}\label{sec:related_work}

This work studies the problem of measuring a Gibbs state non-destructively: either preserving it exactly via detailed balance, or perturbing it in a controlled way that allows efficient recovery.

For non-destructive measurements of Gibbs states, the most relevant prior work is \cite{jiang2026}, which constructs non-destructive measurements for observables commuting with $H$, incurring only logarithmic overhead. For more general observables, \cite{jiang2026} employs the weighted operator Fourier transform to construct a measurement that only slightly disturbs the Gibbs state, under 
additional assumptions. Our Theorem~\ref{thm:informal} generalizes the result of \cite{jiang2026} to arbitrary observables. While the overhead is larger than in the commuting case, it remains modest: scaling logarithmically in both the desired precision and the Hamiltonian evolution time, and linear in $\beta$.
Another approach to 
near-non-destructive measurement is weak measurement~\cite{vaidman2014protectivemeasurementswavefunction},  which, however, can substantially increase the sample complexity for estimating thermal expectation values to a given precision.

Non-destructive measurements have been more extensively studied in the context of ground states. For gapped systems in particular, such measurements can be realized via a variety of techniques, including the weighted operator Fourier transform~\cite{chen2025catalytic}, weak measurement~\cite{vaidman2014protectivemeasurementswavefunction}, 
and the Marriott--Watrous rewinding technique~\cite{marriott2005quantum, farhi2010quantum}. 
A detailed survey of these methods and a comparison of their implementation costs can be found in \cite{chen2025catalytic}.

Another approach to non-destructive measurement is based on recovery maps: one first performs a measurement on the Gibbs state and then applies a recovery procedure to bring the disturbed state back to the Gibbs state. In particular, based on local Markov properties, it has been shown that the disturbance caused by measuring local observables can be recovered by quasi-local quantum channels~\cite{chen2025quantum, kato2025clustering}.  

As presented above, the design of non-destructive measurements is closely related to the design of quantum Gibbs sampling algorithms. Our constructions build on two key ingredients: the smoothing technique via the weighted operator Fourier transform~\cite{chen2025efficientexactnoncommutativequantum} and the discrete-time quantum Metropolis--Hastings framework~\cite{gilyen2026}. Techniques and insights 
from other Gibbs sampling algorithms~\cite{jiang2024weak, Ding_2025, rall2023thermal, temme2011quantum, ding2025end} and related structural results~\cite{bergamaschi2025structural} may further inform the design of new measurement protocols with simpler constructions or reduced implementation costs. Conversely, a fundamentally different approach to 
non-destructive measurement for general observables would itself yield a new Gibbs sampling algorithm, since any such measurement corresponds to a quantum channel admitting the Gibbs state as a fixed point.

\medskip \noindent
\textbf{Concurrent Work.} The concurrent and independent work \cite{chen2026efficient} also employs detailed-balance measurements for measuring quantum thermal states. There, the authors study the task of estimating multiple observables given many copies of the Gibbs state, and show that $\log(M)/\epsilon^2$ samples suffice to estimate $M$ observables to precision $\epsilon$. Our work, in contrast, focuses on estimating a single observable and does not assume access to multiple copies of the thermal state; instead, we begin from a single copy and generate subsequent samples along a single trajectory via non-destructive measurements. The two approaches are complementary: combining our single-trajectory framework with the techniques of \cite{chen2026efficient} could potentially enable the estimation of many observables from a single trajectory.

\section{Preliminaries} \label{sec:preli}
In this section, we introduce some basic concepts used throughout the paper, particularly the quantum Gibbs sampler, the KMS detailed balance condition, the spectral gap, and Heisenberg-picture smoothed observables via the operator Fourier transform.

In this work, we consider a finite-dimensional quantum system associated with a Hilbert space $\mathcal{H} \cong \mathbb{C}^d$, where $d = 2^n$ and $n$ is the number of qubits.  Let $\mathcal{B}(\mathcal{H}) \cong \mathbb{C}^{d \times d}$ denote the algebra of bounded linear operators on $\mathcal{H}$, and let $\mc{D}(\mc{H})$ be the set of quantum states. Given a Hermitian system Hamiltonian 
$H \in \mathcal{B}(\mathcal{H})$ and inverse temperature $\beta > 0$, the quantum Gibbs (thermal) state is given by $\sigma_\beta = e^{-\beta H} / \mathrm{Tr}(e^{-\beta H})$.

A standard approach to preparing quantum thermal states is quantum Gibbs sampling: 
one designs a primitive quantum channel $\mathcal{N}$, efficiently simulable on a 
quantum computer, that admits the Gibbs state $\sigma_\beta$ as its unique fixed 
point~\cite{Chen2025efficient, Ding_2025, chen2025efficientexactnoncommutativequantum, 
gilyen2026}, where
\begin{equation}
    \lim_{k \to \infty} \mathcal{N}^k(\rho) = \sigma_\beta
    \qquad \text{for every initial state } \rho.
\end{equation}
The convergence rate of $\mathcal{N}$ is quantified by the \emph{mixing time},
\begin{align}
    t_{\mathrm{mix}}(\epsilon) := \min\!\Big\{ k \ge 0 \;\Big|\; 
    \sup_{\rho \in \mathcal{D}(\mathcal{H})}
    \|\mathcal{N}^k(\rho) - \sigma_\beta\|_1 \le \epsilon \Big\},
\end{align}
where $\|\cdot\|_1$ denotes the trace norm. Our objective is to \emph{design quantum algorithms for estimating the thermal expectation value $\langle A \rangle_{\sigma_\beta} 
:= \mathrm{Tr}(\sigma_\beta A)$ of a target Hermitian observable $A$ to additive 
error $\epsilon$}. To this end, as previewed in Theorems~\ref{thm:informal} and~\ref{thm:informal_ws}, we develop two measurement protocols for Gibbs states 
and analyze how each reduces the overall sampling cost.

\subsection{Quantum Detailed Balance and the Spectral Gap}  \label{sec:def-chisquare}

One of the key features of existing quantum Gibbs samplers is the KMS detailed balance condition, which facilitates both 
quantum implementation~\cite{chen2025efficientexactnoncommutativequantum, Ding_2025, 
gilyen2026} and convergence analysis~\cite{rouze2025efficient, rouze2026optimal}. 
For any $X, Y \in \mathcal{B}(\mathcal{H})$, the KMS inner product 
is defined by
\begin{align}
    \langle X, Y \rangle_{\sigma_\beta, 1/2}
    := \mathrm{Tr}\!\left(X^\dagger \sigma_\beta^{1/2} Y \sigma_\beta^{1/2}\right).
\end{align}
A superoperator $\mathcal{T}: \mathcal{B}(\mathcal{H}) \to \mathcal{B}(\mathcal{H})$ 
is called a \emph{quantum channel} if it is completely positive and trace-preserving, 
and we denote by $\mathcal{T}^\dagger$ its adjoint with respect to the 
Hilbert--Schmidt inner product $\langle X, Y \rangle_{\mathrm{HS}} := \mathrm{Tr}(X^\dagger Y)$.

\begin{definition}[KMS detailed balance]\label{def:DB}
A superoperator $\mathcal{T}: \mathcal{B}(\mathcal{H}) \to \mathcal{B}(\mathcal{H})$ 
satisfies the KMS detailed balance condition with respect to $\sigma_\beta$ if 
$\mathcal{T}^\dagger$ is self-adjoint under the KMS inner product, i.e., for all 
$X, Y \in \mathcal{B}(\mathcal{H})$,
\begin{align}
    \langle X,\, \mathcal{T}^\dagger(Y) \rangle_{\sigma_\beta, 1/2} 
    = \langle \mathcal{T}^\dagger(X),\, Y \rangle_{\sigma_\beta, 1/2}.
\end{align}
\end{definition}

Any quantum channel $\mathcal{T}$ satisfying $\sigma_\beta$-KMS detailed balance necessarily 
admits $\sigma_\beta$ as a fixed point, $\mathcal{T}(\sigma_\beta) = \sigma_\beta$. 
To bound the mixing time, a standard approach is to analyze the convergence of the 
iterates $\mathcal{T}^k$ via the $\chi^2$-divergence. We first introduce the 
weighted norm
\begin{align}
    \|X\|_{\sigma_\beta,-1/2}^2
    := \mathrm{Tr}\!\left(\sigma_\beta^{-1/2} X^\dagger \sigma_\beta^{-1/2} X\right),
    \qquad X \in \mathcal{B}(\mathcal{H}),
\end{align}
which is closely related to the family of monotone Riemannian metrics in quantum 
information~\cite{petz1996geometries, lesniewski1999monotone}; see 
\cite[Section~2]{chen2025randomized} for the precise connection. The 
$\chi^2$-divergence between a state $\rho$ and the Gibbs state $\sigma_\beta$ 
is then defined by
\begin{align}
    \chi^2(\rho, \sigma_\beta)
    := \|\rho - \sigma_\beta\|_{\sigma_\beta,-1/2}^2
    = \mathrm{Tr}\!\left(
        \sigma_\beta^{-1/2}(\rho - \sigma_\beta)
        \sigma_\beta^{-1/2}(\rho - \sigma_\beta)
    \right).
\end{align}

According to \cite{temme2010chi}, the decay of $\chi^2(\mathcal{N}^k(\rho),\sigma_\beta)$ for the discrete-time quantum Markov chain $\{\mathcal{N}^k\}_{k\ge 0}$ is controlled by the spectral gap.

\begin{definition}
Let $\mathcal{T}$ be a primitive quantum channel satisfying detailed balance with invariant state $\sigma_\beta$, and assume that its spectrum is contained in $[0,1]$\footnote{This can always be enforced, up to a constant-factor loss in the gap, by replacing \(\mathcal T\) with its lazy version \((\mathrm{id}+\mathcal T)/2\).}. 
The \emph{spectral gap} of $\mathcal{T}$ is defined by
\begin{align*}
    \operatorname{gap}(\mathcal{T})
:=
1-
\max_{\substack{X\neq 0\\ \langle X,I\rangle_{\sigma_{\beta},1/2}=0}}
\frac{\langle X,\mathcal{T}^\dagger(X)\rangle_{\sigma_{\beta},1/2}}
     {\langle X,X\rangle_{\sigma_{\beta},1/2}}. 
\end{align*}
\end{definition}

A strictly positive spectral gap, $\lambda := \operatorname{gap}(\mathcal{T})>0$, implies the exponential decay of the $\chi^2$-divergence under iteration: 
\begin{align*}
    \chi^2(\mathcal{T}^k(\rho),\sigma_\beta)
\le
(1 - \lambda)^{2k}\chi^2(\rho,\sigma_\beta),
\end{align*}
which gives, using $\|\rho-\sigma_\beta\|_1^2 \le \chi^2(\rho,\sigma_\beta)$,
\begin{align*}
    \|\mathcal{T}^k(\rho)-\sigma_\beta\|_1 \le
\sqrt{\chi^2(\mathcal{T}^k(\rho),\sigma_\beta)} \le (1- \lambda)^k \sqrt{\chi^2(\rho,\sigma_\beta)}\,.
\end{align*}
Furthermore, the mixing time scales as
\begin{align*}
    t_{\mathrm{mix}}(\epsilon)
=
\mathcal{O}\!\left(
\frac{1}{\lambda}
\log\!\frac{1}{\epsilon\, \sigma_{\beta,\min}}
\right), 
\end{align*}
where $\sigma_{\beta,\min}$ denotes the smallest eigenvalue of $\sigma_\beta$.

\subsection{Smoothing Non-commuting Observables} \label{sec:smoothing}

A central difficulty in extending~\cite{jiang2026} to estimating the thermal expectation 
$\langle A \rangle_{\sigma_\beta}$ for observables $A$ that do not commute with the 
Hamiltonian $H$ is that direct projective measurements severely disturb the Gibbs 
state, and the resulting post-selected states may not average back to $\sigma_\beta$. 
To extract measurement information while maintaining a controlled disturbance of 
$\sigma_\beta$, we use a Heisenberg-picture smoothed observable.

Specifically, following~\cite{chen2025efficientexactnoncommutativequantum, Ding_2025, jiang2026, chen2025catalytic}, for any Hermitian observable $A$, we define the filtered (possibly non-Hermitian) operator $A_f$ associated with an $L^1$-integrable filter function $f: \mathbb{R} \to \mathbb{C}$ by
\begin{align} \label{def:filteredA}
    A_f := \int_{-\infty}^{\infty} f(t)\, e^{iHt} A e^{-iHt}\, \mathrm{d}t.
\end{align}
We assume that $f$ is normalized, i.e.,
\begin{align*}
    \int_{-\infty}^{\infty} f(t)\, \mathrm{d}t = 1,
\end{align*}
which, noting $[\sigma_\beta, e^{iHt}] = 0$, implies that $A_f$ preserves the thermal expectation value of $A$:
\begin{align} \label{eq:equalobs}
    \mathrm{Tr}(\sigma_\beta A_f) = \int_{-\infty}^{\infty} f(t)\, \mathrm{Tr}\!\left(e^{-iHt} \sigma_\beta e^{iHt} A\right) \mathrm{d}t = \int_{-\infty}^{\infty} f(t)\, \mathrm{Tr}(\sigma_\beta A)\, \mathrm{d}t = \mathrm{Tr}(\sigma_\beta A).
\end{align}
This property is crucial for constructing an unbiased estimator of $\langle A \rangle_{\sigma_\beta}$.

A technically essential ingredient in our detailed-balance measurement scheme (\cref{sec:DBmeasure}) and warm-start analysis (\cref{sec:warm-start}) is the controlled behavior of observables under imaginary-time evolution. Although for any fixed finite-dimensional system, the operator $e^{sH} A e^{-sH}$ remains bounded for every finite $s$, its norm can still grow rapidly with the system size (see further discussion in \cite{chen2025quantum, bergamaschi2026fastmixingquantumspin}). In the thermodynamic limit, local observables may even lose boundedness or quasi-local control at finite imaginary time~\cite{bouch2015complex, perezgarcia2023locality}. The purpose of filter smoothing is precisely to tame this growth through a suitably chosen filter function.

\begin{lemma}[Smoothed observable]\label{lem:norm-imaginary}
For a given $s \in \mathbb{R}$, suppose that the Fourier transform $\hat{f}(\omega) = \int_{-\infty}^{\infty} f(t) e^{-i\omega t} \, \mathrm{d}t$ of the filter function satisfies
\begin{align} \label{eq:expdecayf}
    |\hat{f}(\omega)| \leq C e^{-\gamma|\omega|}, \quad \text{for some constant $C > 0$ and $\gamma > |s|$}.
\end{align}
Then the imaginary-time evolution of $A_f$ is again a filtered operator:
\begin{align} \label{eq:shiftfilter}
    e^{sH} A_f e^{-sH} = A_{\tilde{f}},
\end{align}
where $\tilde{f}$ is the analytically continued filter function defined by
\begin{align*}
    \tilde{f}(t) := f(t + is) = \frac{1}{2\pi} \int_{-\infty}^{\infty} \hat{f}(\omega)\, e^{-\omega s + i\omega t}\, \mathrm{d}\omega,
\end{align*}
and assumed to be $L^1$ integrable. Moreover, 
\begin{align} \label{est:norm}
    \|A_{\tilde{f}}\| \leq \|A\| \int_{-\infty}^{\infty} |\tilde{f}(t)|\, \mathrm{d}t.
\end{align}
In particular, for the Gaussian filter $f(t) = \frac{1}{\sqrt{2\pi \tau^2}}e^{-t^2/2\tau^2}$, we have $\|A_{\tilde f}\|\le e^{s^2/(2\tau^2)}\|A\|$.  
\end{lemma}
\begin{proof}
Let $H = \sum_k E_k |k\rangle\langle k|$ 
be the spectral decomposition of $H$, where $E_k$
are the energy eigenvalues and $\ket{k}$ the corresponding energy eigenstates. The matrix elements of $A_f$ are
\begin{align*}
    \langle k | A_f | j \rangle 
    = \int_{-\infty}^{\infty} f(t)\, e^{i(E_k - E_j)t} \langle k | A | j \rangle\, \mathrm{d}t 
    = \hat{f}(E_j - E_k)\, \langle k | A | j \rangle\,.
\end{align*}
Conjugating by $e^{sH}$ and $e^{-sH}$ rescales these elements by $e^{s(E_k - E_j)}$:
\begin{align*}
    \langle k | e^{sH} A_f e^{-sH} | j \rangle 
    = e^{s(E_k - E_j)}\, \hat{f}(E_j - E_k)\, \langle k | A | j \rangle.
\end{align*}
Since $|\hat{f}(\omega)| \leq Ce^{-\gamma|\omega|}$ by assumption and $\gamma > |s|$, 
the product $e^{-s\omega}\hat{f}(\omega)$ is well-defined and absolutely integrable. Moreover, it is bounded for all $\omega = E_j - E_k$. Note that $\tilde{f}(t)$ is defined as the inverse Fourier 
transform of $e^{-s\omega}\hat{f}(\omega)$. We have $\hat{\tilde{f}}(\omega) 
= e^{-s\omega}\hat{f}(\omega)$, and 
\begin{align*}
    \langle k | e^{sH} A_f e^{-sH} | j \rangle 
    = \hat{\tilde{f}}(E_j - E_k)\, \langle k | A | j \rangle 
    = \langle k | A_{\tilde{f}} | j \rangle\,,
\end{align*}
that is, $e^{sH} A_f e^{-sH} = A_{\tilde{f}}$.
The norm bound \eqref{est:norm} follows from the triangle inequality applied to the integral 
definition of $A_{\tilde{f}}$:
\begin{align*}
    \|A_{\tilde{f}}\| 
    \leq \int_{-\infty}^{\infty} |\tilde{f}(t)|\, \|e^{iHt} A e^{-iHt}\|\, \mathrm{d}t 
    \leq \|A\| \int_{-\infty}^{\infty} |\tilde{f}(t)|\, \mathrm{d}t,
\end{align*}
where we used the fact that conjugation by unitaries preserves the operator norm. For the Gaussian filter, $\tilde f(t)=f(t+is)$ satisfies
\[
|\tilde f(t)|=\frac{1}{\sqrt{2\pi\tau^2}}
\exp\!\left(\frac{s^2}{2\tau^2}\right)
\exp\!\left(-\frac{t^2}{2\tau^2}\right),
\]
and therefore $\int_{-\infty}^{\infty}|\tilde f(t)|\,\mathrm{d}t
=\exp(s^2/(2\tau^2))$.
\end{proof}

By standard Paley--Wiener type arguments, the exponential decay condition~\eqref{eq:expdecayf} implies that $f$ extends analytically to the strip $|\Im z|<\gamma$. Thus, \cref{lem:norm-imaginary} shows that, with a sufficiently regular filter function such as a Gaussian, the filtered observable $A_f$ remains well-behaved under imaginary-time evolution $e^{sH}A_f e^{-sH}$, in the sense that its norm can be controlled through the analytically continued filter. This property plays an essential role in constructing detailed-balance measurement operators without rapid norm growth in the low-temperature regime (\cref{sec:DBmeasure}) and in establishing the warm-start property of post-selected states (\cref{sec:warm-start}).

\section{Measurements via Detailed-Balance Channels} \label{sec:DBmeasure}

In this section, we design a measurement channel satisfying the exact KMS detailed balance condition, 
thereby generalizing the single-trajectory Gibbs sampling approach~\cite{jiang2026} to the non-commuting setting and reducing the sampling cost for estimating thermal expectation values.
Specifically, we shall construct a channel that yields an unbiased estimator of $\mathrm{Tr}(\sigma_\beta A)$ for arbitrary observables $A$ not commuting with $H$, while leaving the Gibbs state $\sigma_\beta$ invariant. The construction follows the discrete-time quantum Metropolis-Hastings 
framework of~\cite{gilyen2026}: the channel consists of a jump branch that encodes the desired measurement information and satisfies exact detailed balance, 
and a rejection branch that enforces the trace-preserving property. 

Specifically, our construction proceeds in two main steps. First, we design Kraus operators $K_1$ and $K_2$ in~\eqref{def:k1} and~\eqref{def:k23}, based 
on the smoothed observables $A_f$ and $A_{\tilde{f}}$, whose outcome probabilities encode $\mathrm{Tr}(\sigma_\beta A)$ and which jointly satisfy exact KMS detailed balance (see \cref{lem:DBKraus}).  The imaginary-time shift in the filter function $\tilde{f}$ of $K_2$ is precisely what enforces the detailed balance relation. Since $K_1$ and $K_2$ alone do not define a trace-preserving channel, we follow the quantum Metropolis--Hastings approach of~\cite{gilyen2026} and append a rejection Kraus operator $K$ that restores the trace-preserving property while maintaining exact 
detailed balance. The resulting measurement channel $\mathcal{M}$ is then applied within the single-trajectory framework of~\cite{jiang2026}, where the autocorrelation time $t_{\mathrm{aut}}$ governs the number of steps needed 
between consecutive samples to decorrelate; see \cref{sec:single-traj-DB} below.

Without loss of generality, we assume the target Hermitian observable $A$ 
satisfies $\|A\| \leq 1$. Otherwise, one can replace $A$ by the rescaled observable $A' = A/\|A\|$, in which case, to estimate $\mathrm{Tr}(\sigma_\beta A)$ to additive error $\epsilon$, it suffices to estimate $\mathrm{Tr}(\sigma_\beta A')$ to 
additive error $\epsilon' = \epsilon / \|A\|$.

Our main result of this section is the following theorem. 

\begin{theorem}\label{thm:DB-channel} (Detailed-Balance Measurement and Sample Complexity) Let $A$ be a Hermitian observable with $\|A\| \le 1$. Define the Gaussian filter $f$ and the associated shifted filter $\tilde{f}$ by 
\begin{align*}
    f(t) := \frac{1}{\sqrt{2\pi\tau^2}}\exp\left(- \frac{t^2}{2\tau^2}\right)\,, \quad \tilde{f}(t) := f\left(t - i\frac{\beta}{2} \right).
\end{align*}
For every $\delta\in(0,1)$, one can choose the width $\tau = \cO(\beta+\mathrm{polylog}(\beta\|H\|/\delta))$, the measurement strength $u = \Theta(1)$, and the scaling constant $c = \Theta(1/\log(\beta\|H\|))$ such that the measurement protocol defined by the quantum channel
$$\mathcal{M}[\rho] = K_1 \rho K_1^\dagger + K_2 \rho K_2^\dagger + K \rho K^\dagger,$$
with Kraus operators
$$K_1 = c\sqrt{I+uA_f},\quad K_2 = c\sqrt{I+uA_{\tilde{f}}}, \quad K = \sqrt{\sqrt{\sigma_\beta}(I - K_1^\dagger K_1 - K_2^\dagger K_2)\sqrt{\sigma_\beta}}\sigma_\beta^{-1/2},$$
satisfies the following properties: 
\begin{itemize}
\item \textbf{Exact Detailed Balance}: The channel $\mathcal{M}$ is trace-preserving and satisfies the KMS detailed balance condition with respect to $\sigma_{\beta}$. 
\item \textbf{Efficient Implementation:} 
The channel $\mathcal{M}$, as a quantum instrument, can be simulated with $\delta$ approximation error (in the sense of Definition~\ref{def:approx-instrument}), using 
\begin{align*}
&\cO(\mathrm{polylog}(\beta \|H\|/\delta)) \text{ queries to a block encoding of } A,\\
&\cO\bigl(\beta\,\mathrm{polylog}(\beta\|H\|/\delta)\bigr) \text{ controlled Hamiltonian simulation time for } H, \text{ and}\\
&\cO(\mathrm{polylog}(\beta \|H\|/\delta)) \text{ ancilla qubits.}
\end{align*}
\item \textbf{The Overall Complexity.}
Consider a primitive quantum Gibbs sampling channel $\cN$ satisfying KMS detailed balance, e.g., \cite{Chen2025efficient,Ding_2025,chen2025efficientexactnoncommutativequantum,gilyen2026}.
 To estimate $\mathrm{Tr}(\sigma_\beta A)$ to precision $\epsilon$ with failure probability $\eta$, it suffices to proceed as follows. 
 \begin{enumerate}
     \item[(i)] Apply $t_{\mathrm{mix}}(\eta/3)$ steps of $\cN$ to prepare
  an initial state $\rho_{\mathrm{burn}}$ satisfying $\|\rho_{\mathrm{burn}} - \sigma_\beta\|_1 \le \eta/3$.
    \item[(ii)] Apply the instrument
\(\{\mathcal E_a\}_a\), where \(\mathcal E_a:=\mathcal M\circ \mathcal N\circ \mathcal M_a\), for $$T = \cO\left(\frac{t_{\mathrm{aut},T} \log^2(\beta\|H\|)}{\epsilon^2 \eta}\right)$$
steps, provided the per-step implementation error satisfies $\delta = \cO(\eta/T)$.    
 \end{enumerate} 
 Here $t_{\mathrm{aut},T}$ is the autocorrelation time defined in \eqref{eq:aut} associated with $\{\mc{E}_a\}$. If $\cN$ has spectrum in $[0,1]$ and admits a spectral gap $\lambda > 0$, the autocorrelation time is upper bounded by
\begin{equation} \label{eq:boundtaut}
    t_{\mathrm{aut},T} \le \frac{\cO(1)}{\lambda \log^2(\beta\|H\|)} + \frac{1}{2},
\end{equation}
and hence $
T=\cO((\log^2(\beta\|H\|)+\lambda^{-1})/(\epsilon^2\eta))$. 
\end{itemize}
\end{theorem}

  In what follows, we detail each component of this construction and the resulting protocol, building toward the full proof of Theorem~\ref{thm:DB-channel}. In particular, Section~\ref{sec:DB-conditions} establishes the algebraic condition on $K_1, K_2$ for exact KMS detailed balance and gives their explicit construction. Section~\ref{sec:exact-channel} appends the rejection branch $K$ and verifies the detailed balance and implementation properties. Finally, Section~\ref{sec:single-traj-DB} reviews the single-trajectory sampling protocol in \cite{jiang2026}, sets up the autocorrelation time framework, and shows that the autocorrelation time $t_{\mathrm{aut}}$ controls the decorrelation time between consecutive samples, completing the proof of Theorem~\ref{thm:DB-channel}. We provide a sketch below for completeness.

\begin{proof}[Proof sketch of Theorem~\ref{thm:DB-channel}]
Item~(1) follows from Lemma \ref{lem:DBKraus} and Theorem \ref{thm:rejection}. For item~(2), Lemma~\ref{lem:K1K2-implementation} gives efficient block encodings of the jump branches \(K_1\) and \(K_2\),
while Theorem~\ref{thm:rejection} gives the corresponding implementation of the rejection branch \(K\). Combining these block encodings with Theorem~\ref{thm:approx-channel-be} yields a \(\delta\)-approximate implementation of the full instrument with the stated resource bounds. 
For item~(3), consider the random variable $v$ defined in \eqref{eq:estimator}, which is an unbiased estimator for \(\Tr(\sigma_\beta A)\) and has 
\begin{equation*}
    \Var(v)=\Theta((cu)^{-2})
=\Theta \bigl(\log^2(\beta\|H\|)\bigr).
\end{equation*}
Now apply Lemma~\ref{lem:sample-complexity} to the instrument \(\{\mathcal E_a\}\) after burn-in.
This yields the running time (the number of iterations)
\begin{equation*}
    T=\cO\left(
\frac{\Var(v)}{\epsilon^2\eta}\,
t_{\mathrm{aut},T}
\right)
= \cO\left(
\frac{\log^2(\beta\|H\|)}{\epsilon^2\eta}\,
t_{\mathrm{aut},T}
\right).
\end{equation*}
The stability estimate in Section~\ref{sec:single-traj-DB} and Lemma~\ref{lem:error-accum} shows that choosing \(\delta = \cO(\eta/T)\) keeps the additional implementation error within the failure budget.
Finally, if \(\mathcal N\) has spectrum in \([0,1]\) and spectral gap \(\lambda>0\), then Proposition~\ref{prop:spectral-gap-bound} applies because
\(\mathcal N\), \(\mathcal M\), and the centered measurement map all satisfy KMS detailed balance.
The bound on the parameter \(\theta\) established in Section~\ref{sec:single-traj-DB} gives \eqref{eq:boundtaut}.
\end{proof}

\subsection{Construction of the Jump Branch}\label{sec:DB-conditions}
We begin by defining the unnormalized jump branch of our measurement channel. 
Consider the completely positive map $\mathcal{T}'$ given by two Kraus operators $K_1, K_2 \in \mc{B}(\mc{H})$: 
\begin{align} \label{def:T'}
    \mathcal{T}'[\rho] := K_1\rho K_1^\dagger + K_2\rho K_2^\dagger.
\end{align}
We require $\mathcal{T}'$ to satisfy the KMS detailed balance condition. The following lemma establishes the algebraic condition on $K_1$ and $K_2$ that is sufficient to guarantee this. 

\begin{lemma}\label{lem:DBKraus}
If the Kraus operators $K_1, K_2 \in \mc{B}(\mc{H})$ satisfy
\begin{align} 
    K_2 = \sigma_\beta^{1/2} K_1^\dagger \sigma_\beta^{-1/2}, \label{def:k22}
\end{align}
then the completely positive map $\mathcal{T}'[\rho]$ in \eqref{def:T'} satisfies KMS detailed balance with respect to 
$\sigma_\beta$.
\end{lemma} 
\begin{proof}
By Definition~\ref{def:DB}, $\mathcal{T}'$ satisfies KMS detailed balance if its Hilbert--Schmidt adjoint $\mathcal{T}'^\dagger[X] = K_1^\dagger X K_1 + K_2^\dagger X K_2$ is self-adjoint under the KMS inner product. For $X, Y \in \mathcal{B}(\mathcal{H})$, we compute
\begin{align} \label{eq:innerT}
    \langle X, \mathcal{T}'^\dagger(Y)\rangle_{\sigma_\beta,1/2} 
    = \mathrm{Tr}\!\left(X^\dagger \sigma_\beta^{1/2} K_1^\dagger Y K_1 
    \sigma_\beta^{1/2}\right) 
    + \mathrm{Tr}\!\left(X^\dagger \sigma_\beta^{1/2} K_2^\dagger Y K_2 
    \sigma_\beta^{1/2}\right).
\end{align}
From $K_2 = \sigma_\beta^{1/2} K_1^\dagger \sigma_\beta^{-1/2}$, right-multiplying 
by $\sigma_\beta^{1/2}$ gives $K_2 \sigma_\beta^{1/2} = \sigma_\beta^{1/2} K_1^\dagger$, 
and taking the adjoint yields $\sigma_\beta^{1/2} K_2^\dagger = K_1 \sigma_\beta^{1/2}$. Substituting into \eqref{eq:innerT} gives
\begin{align*}
    \mathrm{Tr}\!\left(X^\dagger (\sigma_\beta^{1/2} K_1^\dagger) Y 
    (K_1 \sigma_\beta^{1/2})\right) 
    = \mathrm{Tr}\!\left(X^\dagger (K_2\sigma_\beta^{1/2}) Y 
    (\sigma_\beta^{1/2} K_2^\dagger)\right),
\end{align*}
and 
\begin{align*}
    \mathrm{Tr}\!\left(X^\dagger (\sigma_\beta^{1/2} K_2^\dagger) Y 
    (K_2 \sigma_\beta^{1/2})\right) 
    = \mathrm{Tr}\!\left(X^\dagger (K_1\sigma_\beta^{1/2}) Y 
    (\sigma_\beta^{1/2} K_1^\dagger)\right).
\end{align*}
Then, summing both identities implies, by \eqref{eq:innerT},
\begin{align*}
    \langle X, \mathcal{T}'^\dagger(Y)\rangle_{\sigma_\beta,1/2} 
    = \langle K_2^\dagger XK_2 + K_1^\dagger XK_1,\, Y\rangle_{\sigma_\beta,1/2} 
    = \langle \mathcal{T}'^\dagger(X),\, Y\rangle_{\sigma_\beta,1/2},
\end{align*}
which establishes the KMS detailed balance of $\mathcal{T}'$. 
\end{proof}

We now explicitly construct Kraus operators $K_1$ and $K_2$ satisfying the algebraic condition~\eqref{def:k22} for KMS detailed balance and admitting efficient block encodings for quantum implementation, using the filtered observable $A_f$ defined in~\eqref{def:filteredA} and its imaginary-time 
shifted counterpart $A_{\tilde{f}}$ defined in~\eqref{eq:shiftfilter}, respectively.

Specifically, for a normalized Hermitian observable $A$ with $\|A\| \le 1$, we employ the standard Gaussian filter 
\begin{equation*}
     f(t) = \frac{1}{\sqrt{2\pi\tau^2}}\exp\left(- \frac{t^2}{2\tau^2}\right) \q \text{with} \q \tau = \Omega(\beta). 
\end{equation*}
We define the first Kraus operator by
\begin{equation}\label{def:k1}
    K_1 := c\sqrt{I + uA_f},
\end{equation}
where $u \in (0, 1)$ is the measurement strength and $c > 0$ is the scaling constant. Since $f$ is real and even, $A_f$ is Hermitian. Moreover, the Gaussian filtered $A_f$ satisfies $\|A_f\| \leq \|A\| \leq 1$, which gives 
$I + uA_f \succ 0$ for $u \in (0,1)$, thus the matrix square root in~\eqref{def:k1} is well-defined, and $K_1$ is Hermitian and positive definite. 

To satisfy the condition~\eqref{def:k22} for KMS detailed balance, we set
\begin{align}\label{def:k2}
    K_2 := \sigma_\beta^{1/2} K_1 \sigma_\beta^{-1/2} 
    = c\,\sigma_\beta^{1/2}\sqrt{I + uA_f}\,\sigma_\beta^{-1/2} 
    = c\sqrt{I + u\sigma_\beta^{1/2} A_f \sigma_\beta^{-1/2}}.
\end{align}
By Lemma~\ref{lem:norm-imaginary}, we have $\sigma_\beta^{1/2} A_f \sigma_\beta^{-1/2} = 
A_{\tilde{f}}$ with $\tilde{f}(t) = f(t - i\beta/2)$. Thus, \eqref{def:k2} 
simplifies to
\begin{align}\label{def:k23}
    K_2 = c\sqrt{I + uA_{\tilde{f}}}.
\end{align}
Since the shifted filter $\tilde{f}$ takes complex values, $A_{\tilde{f}}$ is not Hermitian. Nevertheless, Lemma \ref{lem:norm-imaginary} gives $\|A_{\tilde{f}}\| \leq \exp(\beta^2/(8\tau^2))\|A\|$, which is $\mathcal{O}(1)$ when $\tau = \Omega(\beta)$. Moreover, for any small measurement strength $u \leq 1/(2\|A_{\tilde{f}}\|)$, the operator $I + uA_{\tilde{f}}$ has spectrum bounded away from zero, and then the matrix square root in~\eqref{def:k23} 
is well-defined. The quantum protocol and implementation cost of the block encodings of $K_1$ and $K_2$ are deferred to \cref{app:implementation}. We summarize the construction of $K_1$ and $K_2$ and their implementations in the following lemma,  
which follows by combining Theorem~\ref{thm:Afblock} with Lemma~\ref{lem:qsvt-sqrt} (for $K_1$) and Theorem~\ref{thm:sqrt-be-lcu} (for $K_2$).

\begin{lemma}[Implementation of the jump branches]\label{lem:K1K2-implementation}
Let $U_A$ be an $(\alpha_A,b_A,\epsilon_A)$-block encoding of $A$, where $\|A\|\le 1$ and
$\alpha_A, b_A = \cO(1)$. Assume $\tau=\Omega(\beta)$, choose $0<c\le 1/4$, and choose $u = \cO(1)$ sufficiently small so that the conditions in Lemma \ref{lem:qsvt-sqrt} and Theorem \ref{thm:sqrt-be-lcu} hold. Then, for any sufficiently small target accuracy $\epsilon \in (0,1)$, if
\begin{equation} \label{assumpa}
    \epsilon_A = \cO(\epsilon^2/\log^2(1/\epsilon)), 
\end{equation}
then one can construct constant-normalization block encodings of $K_1$ and $K_2$ with error $\cO(\epsilon)$. The construction uses
\begin{align*}
&\cO(\log(1/\epsilon)) \text{ queries to } U_A,\\
&\cO\,\bigl(\tau\,\mathrm{polylog}(1/\epsilon)\bigr) \text{ controlled Hamiltonian simulation time, and}\\
&\cO \bigl(\mathrm{polylog}(\tau\|H\|/\epsilon)\bigr) \text{ ancilla qubits in total.}
\end{align*}
\end{lemma}

\begin{proof}
By Theorem~\ref{thm:Afblock}, when $\tau = \Omega(\beta)$, for any $g\in\{f,\tilde f\}$ and $\xi_g > 0$, one can construct an $(\alpha_g,b_g,\delta_g)$-block encoding of $A_g$ such that
\begin{equation} \label{auxeqlemm1}
    \alpha_g = \cO(1)\,, \q b_g = \cO\,\bigl(\log(\tau\|H\|)+\log\log(1/\xi_g)\bigr)\,,\q \delta_g = \cO(\epsilon_A + \xi_g)\,,
\end{equation}
using one query to $U_A$ and controlled Hamiltonian simulation time $T_{A_g} = \cO\,\bigl(\tau\sqrt{\log(1/\xi_g)}\bigr)$. 
By increasing the constant normalization if necessary, we take $\alpha_g=\Theta(1)$ and $\|A_g\|\le \alpha_g$ for both $g=f,\tilde f$.

We first consider $K_1$ and choose $\xi_f = \Theta(\epsilon^2/\log^2(1/\epsilon)) < 1$. Applying Lemma \ref{lem:qsvt-sqrt} to $B=A_f$, with target polynomial approximation error $\cO(\epsilon)$, there exists a block encoding of $K_1=c\sqrt{I+uA_f}$ with error bounded, using $\alpha_f = \Theta(1)$ and $d = \cO(\log(1/\epsilon))$, by
\begin{equation} \label{auxeqlemm2}
    4d\sqrt{\delta_f/\alpha_f}+\frac{\epsilon}{2} = \cO(\epsilon + \log(1/\epsilon) \sqrt{\delta_f}) = \cO(\epsilon + \log(1/\epsilon) \sqrt{\epsilon_A}), 
\end{equation}
where the last equality is by $\delta_f = \cO(\epsilon_A + \epsilon^2/\log^2(1/\epsilon))$ from \eqref{auxeqlemm1}. Under the assumption $\epsilon_A = \cO(\epsilon^2/\log^2(1/\epsilon))$, the error \eqref{auxeqlemm2} reduces to $\cO(\epsilon)$. In this construction, the query complexity is $\cO(d) = \cO(\log(1/\epsilon))$ uses of the block encoding of $A_f$, and hence $\cO(\log(1/\epsilon))$ queries to $U_A$. Then the total controlled Hamiltonian simulation time is
\begin{equation*}
    \cO( d\cdot T_{A_f}) = \cO\,\bigl(d\cdot \tau\sqrt{\log(1/\xi_f)}\bigr) = \cO\,\bigl(\tau\,\mathrm{polylog}(1/\epsilon)\bigr),
\end{equation*}
and the ancilla count is $b_f+ \cO(1)$ with $b_f$ in \eqref{auxeqlemm1}.

We next consider $K_2$ and choose $\xi_{\tilde f} = \Theta(\epsilon)$ in \eqref{auxeqlemm1}. 
Set $r_{\tilde f}:=|u|\alpha_{\tilde f}$. By the choice of $u$, $r_{\tilde f}\le 1/2$ and, for sufficiently small $\epsilon$, $|u|(\alpha_{\tilde f}+\delta_{\tilde f})<1$.
Applying Theorem \ref{thm:sqrt-be-lcu} to $B=A_{\tilde f}$ with Taylor series truncation error $\cO(\epsilon)$, it follows that there exists a $(\gamma_2,b',\cO(\epsilon))$-block encoding of
$K_2=c\sqrt{I+uA_{\tilde f}}$, where $\gamma_2=c(2-\sqrt{1-r_{\tilde f}})=\cO(1)$.
Indeed, the error is $\cO(\epsilon+\delta_{\tilde f})=\cO(\epsilon+\epsilon_A)=\cO(\epsilon)$ by \eqref{assumpa}. The number of queries to the block encoding of $A_{\tilde f}$ is
$\cO(\log(1/\epsilon)/\log(1/r_{\tilde f}))=\cO(\log(1/\epsilon))$, hence again $\cO(\log(1/\epsilon))$
queries to $U_A$. Similarly, the controlled Hamiltonian simulation time is $\cO\,\bigl(\tau\,\mathrm{polylog}(1/\epsilon)\bigr)$. 
Finally, Theorem \ref{thm:sqrt-be-lcu} also directly gives the ancilla count
\begin{equation*}
    \cO\left(\frac{b_{\tilde f}\log(1/\epsilon)}{\log(1/r_{\tilde f})}\right)
= \cO \bigl(\mathrm{polylog}(\tau\|H\|/\epsilon)\bigr).    
\end{equation*}
The proof is complete. 
\end{proof}

\medskip \noindent
\textbf{Extracting measurement information from $K_1$ and $K_2$.}
We now show that observing either classical outcome associated with $K_1$ or $K_2$ provides equivalent, unbiased information about $\mathrm{Tr}(\sigma_\beta A)$.
First, the probability $p_1$ of obtaining the outcome corresponding to the 
$K_1$-branch with respect to the Gibbs state is
\begin{align*}
    p_1 = \mathrm{Tr}(\sigma_\beta K_1^\dagger K_1) 
    = \mathrm{Tr}(\sigma_\beta K_1^2) 
    = c^2\,\mathrm{Tr}(\sigma_\beta(I + uA_f)) 
    = c^2(1 + u\,\mathrm{Tr}(\sigma_\beta A)),
\end{align*}
where we used $\mathrm{Tr}(\sigma_\beta A_f) = \mathrm{Tr}(\sigma_\beta A)$ 
from~\eqref{eq:equalobs}. This immediately yields an unbiased estimator for the thermal expectation $\mathrm{Tr}(\sigma_\beta A)$. Similarly, using $K_2 = \sigma_\beta^{1/2} K_1  \sigma_\beta^{-1/2}$, the probability $p_2$ corresponding to the $K_2$-branch is
\begin{align*}
    p_2 = \mathrm{Tr}(\sigma_\beta K_2^\dagger K_2) 
    = \mathrm{Tr}\!\left(\sigma_\beta\,
    (\sigma_\beta^{-1/2} K_1 \sigma_\beta^{1/2})
    (\sigma_\beta^{1/2} K_1 \sigma_\beta^{-1/2})\right) 
    = \mathrm{Tr}(\sigma_\beta K_1^2) = p_1.
\end{align*}
Thus, both branches yield the same outcome probability, with the signal term 
$u\,\mathrm{Tr}(\sigma_\beta A)$ scaling linearly with the measurement strength $u$. 

Assigning outcomes $z=1$ to branches $K_1$ and $K_2$, and $z=0$ to the rejection branch $K$, the estimator
\begin{equation} \label{eq:estimator}
    v:=\frac{z}{2c^2u}-\frac{1}{u}
\end{equation}
satisfies $\mathbb{E}[v]=\Tr(\sigma_\beta A)$ when the input state is $\sigma_\beta$ and the measurement channel is implemented exactly. Moreover, a direct computation gives 
\begin{equation*}
    \Var(v)=\Theta \bigl((cu)^{-2}\bigr).
\end{equation*}
In our construction, we choose $u=\Theta(1)$ as a sufficiently small absolute constant and
\[
c=\Theta(1/\log(\beta\|H\|))
\]
(see Section~\ref{sec:exact-channel}), so the variance is
\begin{equation} \label{eq:variancem}
    \Var(v)=\Theta\bigl(\log^2(\beta\|H\|)\bigr).
\end{equation}

\subsection{Construction of the Rejection Branch}\label{sec:exact-channel}

To complete the measurement channel, we append a rejection branch 
$\mathcal{R}[\rho] = K\rho K^\dagger$ to compensate for the non-trace-preserving 
property of $\mathcal{T}'$, noting that, once $u$ is fixed as a sufficiently
small absolute constant, i.e., $u \le 1/(2 \max\{\|A_f\|, \|A_{\tilde{f}}\|\})$, 
\begin{equation} \label{eq:boundtio}
    \mathcal{T}'^\dagger[I] = K_1^\dagger K_1 + K_2^\dagger K_2 \leq 3c^2 I.
\end{equation}
The construction of such a rejection branch, which maintains the KMS detailed balance condition and efficient quantum implementability, follows the discrete-time quantum Metropolis-Hastings sampler framework of~\cite{gilyen2026}. 

\begin{theorem} \label{thm:rejection} (Trace-Preserving Completion and Efficient Implementation \cite{gilyen2026}). Consider any completely positive map $\mathcal{T}'$ satisfying $\sigma_\beta$-KMS detailed balance such that $\mathcal{T}'^\dagger[I] \le I$. The completely positive map defined as
$$\mathcal{M} = \mathcal{T}' + \mathcal{R}, \quad \text{where } \mathcal{R}[\cdot] = K[\cdot]K^\dagger$$
satisfies $\sigma_\beta$-KMS detailed balance and is trace-preserving, given the rejection Kraus operator:
$$K := \sqrt{\sqrt{\sigma_\beta}(I - \mathcal{T}'^\dagger[I])\sqrt{\sigma_\beta}}\sigma_\beta^{-1/2}.$$

Moreover, suppose $\mathcal{T}'[\cdot] = K_1 \cdot K_1^\dag + K_2\cdot K_2^\dag$, where $K_1$ and $K_2$ are the jump-branch Kraus operators constructed in Section \ref{sec:DB-conditions} with appropriate parameters
$$c = \mathcal{O}(1/\log(\beta\|H\|)),\quad u = \cO(1),\quad \tau = \mathcal{O}(\beta +\mathrm{polylog}(\beta\|H\|/\varepsilon')).$$
Then, for any target accuracy $\varepsilon'\in(0,1/4]$, the operator $K$ admits a normalization-$2$, $\varepsilon'$-approximate block encoding using
\begin{align*}
&\cO\bigl(\operatorname{polylog}(\beta\|H\|/\varepsilon')\bigr)
\text{ queries to a block encoding of } O = \mc{T}^{\prime\dagger}[I],\\
&\cO \bigl(\beta\|H\|\operatorname{polylog}(\beta\|H\|/\varepsilon')\bigr)
\text{ queries to a block encoding of } H, \text{ and}\\
&\cO\bigl(\operatorname{polylog}(\beta\|H\|/\varepsilon')\bigr) \text{ ancilla qubits.}
\end{align*}
In addition, the block encoding of
$
O=K_1^\dagger K_1+K_2^\dagger K_2
$
can be implemented using
\begin{align*}
&\cO\bigl(\operatorname{polylog}(\beta\|H\|/\varepsilon')\bigr)
\text{ queries to a block encoding of } A, \text{ and}\\
&\cO\bigl(\beta\,\operatorname{polylog}(\beta\|H\|/\varepsilon')\bigr)
\text{ controlled Hamiltonian simulation time for } H.
\end{align*}
\end{theorem}

The key insight behind the efficient implementation is that \cite{gilyen2026} shows the operator $K$ admits a series expansion in which each term can be written as an integral of short-time Hamiltonian evolutions. The scaling $c = \mathcal{O}(1/\log(\beta\|H\|))$ is chosen to make such series converge. When $O$ satisfies the required quasi-local-in-energy condition, this expansion enables an efficient quantum implementation. The proof of Theorem~\ref{thm:rejection} is deferred to Appendix~\ref{sec:implementR}; we refer the reader to \cite{gilyen2026} for full details.

\begin{remark}
The implementation cost in Theorem~\ref{thm:rejection} involves $\tilde{\mathcal{O}}(\beta \|H\|)$ queries to a block encoding of $H/\alpha$ with $\alpha = \Theta(\|H\|)$. 
This matches the standard cost of block-encoding- and QSVT-based Hamiltonian simulation \cite{Low_2019, Gily_n_2019} for simulating $H$ up to time $\tilde{\mathcal{O}}(\beta)$. We therefore absorb this cost into the Hamiltonian simulation cost reported in Theorem~\ref{thm:DB-channel}.
\end{remark}

Having established the block encodings of $K_1$, $K_2$, and $K$ individually, we discuss in Appendix~\ref{app:implementation-POVM} how to combine these components to obtain an $\epsilon'$-approximate implementation of the full measurement channel $\mathcal{M}$.

\subsection{The Single-Trajectory Protocol}\label{sec:single-traj-DB}

With the detailed-balance measurement channel $\mathcal{M}$ constructed in 
Sections~\ref{sec:DB-conditions} and~\ref{sec:exact-channel}, we now describe 
how to combine it with a Gibbs sampling channel $\mathcal{N}$ within the 
single-trajectory framework of \cite{jiang2026} to estimate $\mathrm{Tr}(\sigma_\beta A)$.

\medskip \noindent
\textbf{Protocol.} Given a Gibbs sampling channel $\mathcal{N}$ satisfying
$\sigma_\beta$-KMS detailed balance~\cite{chen2025efficientexactnoncommutativequantum, Ding_2025}, 
the measurement channel $\mathcal{M}$ from Theorem~\ref{thm:DB-channel}, an initial 
state $\rho$, and integers $T_{\mathrm{burn}}$ and $T$, the protocol proceeds as 
follows.
\begin{itemize}
    \item \textbf{Burn-in.} Apply $\mathcal{N}$ for $T_{\mathrm{burn}}$ steps 
    to $\rho$ to obtain a state close to $\sigma_\beta$.
    
    \item \textbf{Sampling.} For each step $t = 1, \dots, T$, apply the composite 
    sandwich channel $\mathcal{E} := \mathcal{M} \circ \mathcal{N} \circ \mathcal{M}$.
    
    \item \textbf{Outcome recording.} Let $a_t \in \{1, 2, 3\}$ denote the branch 
    label of the first (rightmost) application of $\mathcal{M}$ at step $t$, with 
    associated real-valued outcome $v_{a_t}$, where $v_1 = v_2 = \frac{1}{2c^2 u}$ 
    and $v_3 = 0$.
    
    \item \textbf{Estimation.} Output $X_T - \frac{1}{u}$, where $X_T := 
    \frac{1}{T}\sum_{t=1}^T v_{a_t}$, as the estimator of $\mathrm{Tr}(\sigma_\beta A)$.
\end{itemize}

\medskip \noindent
\textbf{Sample complexity via the autocorrelation time.}
Since both $\mathcal{N}$ and $\mathcal{M}$ satisfy exact KMS detailed balance 
with respect to $\sigma_\beta$, the composite channel $\mathcal{E}$ also admits 
$\sigma_\beta$ as its invariant state. A self-contained treatment of the measurement 
outcome statistics and the relevant operator-theoretic quantities is given in 
Appendix~\ref{app:measurement-stats}; here we summarize the key definitions and results.

Let $\{K_a\}_{a=1,2,3}$ be the Kraus operators of $\mathcal{M}$ (with $K_3 = K$ the rejection branch), and let $v_a \in \mathbb{R}$ denote the real-valued outcome associated with branch $a$. Note that the statistics of $\mathcal{M}$ depend on the choice of outcome function $v$; here we use $v_1 = v_2 = \frac{1}{2c^2u}$ and $v_3 = 0$ as defined in the protocol. Since the post-processing channels $\mathcal M$ and $\mathcal N$ are trace-preserving, $\Tr(\mathcal E_a(\sigma_\beta))=\Tr(K_a\sigma_\beta K_a^\dagger)$. Thus the mean and variance of the recorded outcome in the stationary state are:
$$\bar{v}=\mathbb{E}_{\sigma_\beta}(\mathcal{M}) := \sum_a v_a\, \mathrm{Tr}(K_a \sigma_\beta K_a^\dagger), \qquad \mathbb{V}\mathrm{ar}_{\sigma_\beta}(\mathcal{M}) := \sum_a (v_a - \bar{v})^2\, \mathrm{Tr}(K_a \sigma_\beta K_a^\dagger).$$
The statistical performance of the empirical average $X_T$ is governed by the integrated autocorrelation time of $\mathcal{E}$:

\begin{align}
t_{\mathrm{aut}, T} := \frac{1}{2} + \sum_{t=1}^{T} \left(1 - \frac{t}{T}\right) \frac{\mathbb{C}\mathrm{orr}_{\sigma_\beta}(\mathcal{E}^{t-1})}{\mathbb{V}\mathrm{ar}_{\sigma_\beta}(\mathcal{M})}, \label{eq:aut}
\end{align}
where the autocorrelation function is $\mathbb{C}\mathrm{orr}_{\sigma_\beta}(\mathcal{E}^t) := \mathrm{Tr}(\widehat{\mathcal{E}}_v \circ \mathcal{E}^t \circ \widehat{\mathcal{E}}_v(\sigma_\beta))$ and $\widehat{\mathcal{E}}_v = \sum_a (v_a - \bar{v})\mathcal{E}_a$. 
Here, $\widehat{\mathcal{M}}(X) := \sum_a (v_a - \bar{v}) K_a X K_a^\dagger$ is the centered measurement map. By \cite{jiang2026} (see Lemma \ref{lem:sample-complexity} for the precise statement), set $T_{\mathrm{burn}}:=t_{\rm mix}(\frac{1}{3}\eta)$. After the burn-in stage, the current state $\rho_{\mathrm{burn}}$ satisfies $\|\rho_{\mathrm{burn}} -\sigma_\beta \|_1\le \frac{1}{3}\eta$. Then in the sampling stage, it suffices to set
$$T = \mathcal{O}\!\left(\frac{\mathbb{V}\mathrm{ar}_{\sigma_\beta}(\mathcal{M})}{\epsilon^2 \eta}\, t_{\mathrm{aut}, T}\right)$$
to estimate $\mathrm{Tr}(\sigma_\beta A)$ to precision $\epsilon$ with failure probability $2\eta/3$ for the exact implementation. Since $\mathbb{V}\mathrm{ar}_{\sigma_\beta}(\mathcal{M}) = \Theta(\log^2(\beta\|H\|))$ by \eqref{eq:variancem}, the total number of required steps is:
\begin{align} \label{eq:T} T = \mathcal{O}\!\left(\frac{\log^2(\beta\|H\|)}{\epsilon^2 \eta}\, t_{\mathrm{aut},T}\right).\end{align}

\medskip \noindent
\textbf{Bounding the autocorrelation time via the spectral gap.}
When the Gibbs sampler $\mathcal{N}$ has spectrum in $[0,1]$ and spectral gap $\lambda > 0$, the autocorrelation time can be bounded by $\mathcal{O}(1/\lambda)$. More specifically, by \cite{jiang2026} (see Proposition~\ref{prop:spectral-gap-bound} for the precise statement), if both $\mathcal{N}$ and $\mathcal{M}$ satisfy KMS detailed balance and the centered measurement map $\widehat{\mathcal{M}}$ also satisfies KMS detailed balance, then:
\begin{equation} \label{eq:boundaut}
    t_{\mathrm{aut}, T} \le \frac{\theta}{\lambda} + \frac{1}{2}, \quad \text{where} \quad \theta := \frac{\sum_{a,b}(v_a - \bar{v})(v_b - \bar{v})\mathrm{Tr}(K_b K_a \sigma_\beta K_a^\dagger K_b^\dagger)}{\mathbb{V}\mathrm{ar}_{\sigma_\beta}(\mathcal{M})}.
\end{equation}
We now verify that these conditions hold in our construction and that $\theta = \mathcal{O}(1/\log^2(\beta\|H\|))$.

\medskip \noindent
\textbf{Detailed balance of $\widehat{\mathcal{M}}$.} Recall that $\widehat{\mathcal{M}}(X) = \sum_a (v_a - \bar{v}) K_a X K_a^\dagger$. Since $v_1 = v_2 = \frac{1}{2c^2u}$ and $v_3 = 0$, the centered coefficients are $v_1 - \bar{v} = v_2 - \bar{v} = \frac{1}{2c^2u} - \bar{v}$ and $v_3 - \bar{v} = -\bar{v}$. Thus:
$$\widehat{\mathcal{M}}(X) = \left(\tfrac{1}{2c^2u}-\bar{v}\right)\left(K_1 X K_1^\dagger + K_2 X K_2^\dagger\right) - \bar{v}\, K X K^\dagger = \left(\tfrac{1}{2c^2u}-\bar{v}\right)\mathcal{T}'[X] - \bar{v}\,\mathcal{R}[X].$$
Both $\mathcal{T}'$ and $\mathcal{R}$ individually satisfy KMS detailed balance: $\mathcal{T}'$ by Lemma~\ref{lem:DBKraus}, and $\mathcal{R}$ by Theorem~\ref{thm:rejection}. Since $\widehat{\mathcal{M}}$ is a linear combination of two maps satisfying KMS detailed balance, $\widehat{\mathcal{M}}$ also satisfies KMS detailed balance with respect to $\sigma_\beta$.

\medskip \noindent
\textbf{Bounding $\theta$.} The centered outcomes are $v_1 - \bar{v} = v_2 - \bar{v} = \frac{1}{2c^2u} - \bar{v} = \mathcal{O}(1/(c^2u))$ and $v_3 - \bar{v} = -\bar{v} = \mathcal{O}(1/u)$. Using $\mathrm{Tr}(K_b K_a \sigma_\beta K_a^\dagger K_b^\dagger) \le \|K_b\|^2\|K_a\|^2$ and $\|K_1\|,\|K_2\| = \mathcal{O}(c)$, $\|K_3\| \le 1$, the numerator in \eqref{eq:boundaut} is bounded by $\mathcal{O}(1)$. Combining this with \eqref{eq:variancem} gives $\theta = \mathcal{O}(1/\log^2(\beta\|H\|))$. Thus we have
$$
t_{\mathrm{aut},T} \le \frac{\mathcal{O}(1)}{\lambda \log^2(\beta\|H\|) } + \frac{1}{2}.
$$
Combining this with \eqref{eq:T}, the total number of required steps is $T = \mathcal{O}\left(\frac{1}{\epsilon^2 \eta}(\log^2(\beta \|H\|) + \lambda^{-1}) \right)$.

\medskip \noindent
\textbf{Stability analysis for the approximate implementation.}
In practice, the sandwich channel $\mathcal{E} = \mathcal{M} \circ \mathcal{N} 
\circ \mathcal{M}$ is only implemented approximately. We model this within the 
quantum instrument framework of Appendix~\ref{app:measurement-stats}. Each step 
of the ideal protocol defines a quantum instrument $\{\mathcal{E}_a\}_{a=1,2,3}$, 
where each branch $\mathcal{E}_a(\rho) = (\mathcal{M} \circ \mathcal{N})(K_a \rho K_a^\dagger)$ 
is indexed by the Kraus operator $K_a$ and carries real-valued outcome $v_a$. A 
trajectory of branch labels $\vec{a} = (a_1, \dots, a_T) \in \{1,2,3\}^T$ 
determines the sequence of outcomes $(v_{a_1}, \dots, v_{a_T})$, and the ideal 
trajectory distribution is induced by recursively applying $\{\mathcal{E}_a\}$ 
to the post-burn-in state $\rho_{\mathrm{burn}}$:
\begin{align}
    \mathbb{P}(\vec{a}) = \mathrm{Tr}\!\left(
    \mathcal{E}_{a_T} \circ \cdots \circ \mathcal{E}_{a_1}(\rho_{\mathrm{burn}})
    \right).
\end{align}
In the approximate protocol, we instead apply an $\epsilon'$-approximate instrument 
$\{\widetilde{\mathcal{E}}_a\}$ at each step (in the sense of 
Definition~\ref{def:approx-instrument}), inducing the trajectory distribution
\begin{align}
    \widetilde{\mathbb{P}}(\vec{a}) = \mathrm{Tr}\!\left(
    \widetilde{\mathcal{E}}_{a_T} \circ \cdots \circ 
    \widetilde{\mathcal{E}}_{a_1}(\rho_{\mathrm{burn}})
    \right).
\end{align}
By Lemma~\ref{lem:error-accum}, the total variation distance between the two 
trajectory distributions satisfies
\begin{align}
    \|\mathbb{P} - \widetilde{\mathbb{P}}\|_{\mathrm{TV}} \leq \tfrac{1}{2} T\epsilon'.
\end{align}
We now translate this bound into a failure probability guarantee. Define the failure 
event $S := \{\vec{a} : |X_T - \frac{1}{u} - \mathrm{Tr}(\sigma_\beta A)| \geq \epsilon\}$, 
where $X_T = \frac{1}{T}\sum_{t=1}^T v_{a_t}$ is the empirical average. By the 
definition of total variation distance and Lemma~\ref{lem:sample-complexity}, which 
gives $\mathbb{P}(S) \leq \frac{2}{3}\eta$ for the exact protocol, we obtain
\begin{align}
    \widetilde{\mathbb{P}}(S) \leq \mathbb{P}(S) + \|\mathbb{P} - 
    \widetilde{\mathbb{P}}\|_{\mathrm{TV}} \leq \frac{2}{3}\eta + \frac{1}{2}T\epsilon'.
\end{align}
Setting $\epsilon' = \mathcal{O}(\eta/T)$ makes the second term at most $\eta/3$, 
giving total failure probability $\eta$. Substituting 
$T = \mathcal{O}\!\left(\frac{\log^2(\beta\|H\|)}{\epsilon^2\eta}\,t_{\mathrm{aut},T}\right)$, the required per-step accuracy is
\begin{align}
    \epsilon' = \mathcal{O}\!\left(
    \frac{\epsilon^2\eta^2}{\log^2(\beta\|H\|) \cdot t_{\mathrm{aut},T}}
    \right).
\end{align}
Since the implementation costs of $\mathcal{M}$ (Lemma~\ref{lem:K1K2-implementation}, Theorem~\ref{thm:rejection}, and Theorem~\ref{thm:approx-channel-be}) 
and $\mathcal{N}$~\cite{chen2025efficientexactnoncommutativequantum, Ding_2025} 
scale logarithmically in $\epsilon'$, the approximation overhead is only a 
polylogarithmic additive factor in the overall complexity.

\section{Measurement-Remixing via Warm-Start in \texorpdfstring{$\chi^2$}{chi-squared}-Divergence} \label{sec:warm-start}

Our second construction is simpler than the detailed-balance approach of
Section~\ref{sec:DBmeasure}, at the cost of requiring a positive spectral gap
$\lambda > 0$. We design a two-outcome positive operator-valued measure (POVM) channel based entirely on the real-time smoothed observable $A_f$, whose outcomes yield an unbiased estimator
of $\mathrm{Tr}(\sigma_\beta A)$ when the input state is $\sigma_\beta$. Since the filter $f$ is real-valued, the
construction requires neither an imaginary-time shift nor a rejection branch, 
making it considerably simpler to implement than the exact detailed-balance channel.

The key analytical property is that, despite not satisfying detailed balance, the 
post-selected state maintains a bounded $\chi^2$-divergence from $\sigma_\beta$. 
Under a spectral gap $\lambda > 0$, this warm-start guarantee ensures that the 
post-selected state returns to within $\chi^2$-divergence (hence also trace-norm) $\epsilon$-proximity of $\sigma_\beta$ 
in time $\mathcal{O}(\lambda^{-1}\log(1/\epsilon))$, rather than the full mixing 
time. Crucially, the per-sample remixing cost is entirely independent of 
$\sigma_{\beta,\mathrm{min}}$: the factor $\log\,(\sigma_{\beta,\mathrm{min}}^{-1})$ 
that inflates the global mixing time appears only in the one-time burn-in cost 
and not in the per-sample cost.

We assume $\|A\| \le 1$ without loss of generality. The main result is as follows.

\begin{theorem}[Measurement-Remixing Complexity]\label{thm:complexity1}
For any Hermitian observable $A$ with $\|A\| \le 1$, consider the real Gaussian filter $f(t) = \frac{1}{\sqrt{2\pi\tau^2}}\exp\!\left(-\frac{t^2}{2\tau^2}\right)$. One can choose $\tau = \Theta(\beta)$ and a sufficiently small constant measurement strength $u = \Theta(1)$ such that the measurement channel
$$\mathcal{M}[\rho] = K_+ \rho K_+^\dagger + K_- \rho K_-^\dagger, \quad K_\pm = \sqrt{\frac{I \pm u A_f}{2}},$$
satisfies the following. Suppose \(\mc{N}\) is a primitive quantum Gibbs sampling channel with unique fixed
point \(\sigma_\beta\), satisfying KMS detailed balance, with spectrum contained
in \([0,1]\), and spectral gap \(\lambda>0\). Then:
\begin{itemize}
    \item \textbf{Warm Start and Remixing Cost:} Let $\rho_\pm'$ denote the post-selected state after applying $\mathcal{M}$ to $\rho$ and recording outcome $\pm$, i.e.,
    $\rho_\pm' = K_\pm \rho K_\pm^\dagger / \mathrm{Tr}(\rho K_\pm^\dagger K_\pm)$. If $\chi^2(\rho, \sigma_\beta) = \mathcal{O}(1)$, then $\chi^2(\rho_\pm', \sigma_\beta) = \mathcal{O}(1)$. Moreover, $k_0 = \mathcal{O}\!\left(\frac{1}{\lambda}\log\frac{1}{\epsilon}\right)$ steps of $\mathcal{N}$ suffice to bring the post-selected state back to $\chi^2$-divergence $\epsilon$-proximity of the Gibbs state:
    $\chi^2(\mathcal{N}^{k_0}(\rho_\pm'), \sigma_\beta)\le \epsilon$.
    \item \textbf{Efficient Implementation:} An $\epsilon'$-approximate implementation of $\mathcal{M}$ can be realized with
\begin{align*}
&\mathcal{O}(\mathrm{polylog}(1/\epsilon')) \text{ queries to a block encoding of } A,\\
&\mathcal{O}(\beta\mathrm{polylog}(1/\epsilon')) \text{ Hamiltonian simulation time for } H, \text{ and}\\
&\mathcal{O}(\log(\beta \|H\|) + \log\log(1/\epsilon')) \text{ ancilla qubits.}
\end{align*}
    \item \textbf{Sample Complexity:} To estimate $\mathrm{Tr}(\sigma_\beta A)$ to accuracy $\epsilon$ with failure probability $\eta$, starting from a burn-in state $\rho_{\mathrm{burn}}$ satisfying $\|\rho_{\mathrm{burn}} - \sigma_\beta\|_1 \le \eta/3$, the total number of required measurement-remixing steps is:
    $$T = \mathcal{O}\!\left(\frac{\log(1/\eta)}{\epsilon^2}\right),$$
    provided each combined measurement-remixing step (applying $\mathcal{M}$ followed by $\mathcal{N}^{k_0}$) is implemented to per-step accuracy
    $$\epsilon' = \mathcal{O}\!\left(\frac{\epsilon^2\eta}{\log(1/\eta)}\right)$$
    in the sense of Definition~\ref{def:approx-instrument}.
\end{itemize}
\end{theorem}
Compared to the detailed-balance approach of Section~\ref{sec:DBmeasure}, this construction is more implementation-friendly:
\begin{itemize}
    \item \textbf{Fewer ancilla qubits.} On top of the ancillas needed for a block encoding of $A_f$, implementing $K_\pm$ via QSVT requires only a constant number of additional ancilla qubits, since the target function $\sqrt{(1\pm ux)/2}$ admits a real polynomial approximation (see \cite{Gily_n_2019} and Section~\ref{app:QSVT}). By contrast, Section~\ref{sec:DBmeasure} requires $\mathrm{polylog}(\beta \|H\|/\epsilon)$ additional ancillas for the general (non-Hermitian) eigenvalue transformation implementing $K_2$, plus further ancillas for the rejection operator $\mathcal{R}$ (Appendix~\ref{sec:implementR}).
    \item \textbf{No rejection branch.} The rejection branch $K$ in Section~\ref{sec:DBmeasure} is costly in two ways: it contributes zero signal but increases the estimator variance by a $\mathrm{polylog}(\beta\|H\|)$ factor, and its implementation via $\mathcal{R}$ requires an additional $\mathrm{polylog}(\beta\|H\|)$ overhead. Neither cost is present here.
    \item \textbf{Better failure probability dependence.} Because the measurement outcomes are bounded, martingale concentration yields sample complexity $T = \mathcal{O}(\epsilon^{-2}\log(1/\eta))$, with only logarithmic dependence on the failure probability $\eta$. By contrast, the detailed-balance approach relies on Chebyshev's inequality, giving a $1/\eta$ dependence. We note that a Markov chain concentration inequality~\cite{Girotti_2023} could potentially improve the $\eta^{-1}$ dependence in Theorem~\ref{thm:DB-channel} to $\log \eta^{-1}$. 
\end{itemize}

In what follows, we detail the construction and analysis, building toward the full proof of Theorem~\ref{thm:complexity1}. Section~\ref{sec:POVM-scheme} details the explicit construction of the measurement channel $\mathcal{M}$. Section~\ref{sec:warm-start-proof} establishes the $\chi^2$-warm-start guarantee (Theorem~\ref{thm:warmstart}). Section~\ref{sec:remixing-complexity} analyzes the sample complexity for estimating $\mathrm{Tr}(\sigma_\beta A)$ from the measurement outcomes.

\subsection{Construction of the Measurement Channel}\label{sec:POVM-scheme}
We use the same real Gaussian filter $f(t) = \frac{1}{\sqrt{2\pi\tau^2}}\exp(-t^2/2\tau^2)$ 
as in Section~\ref{sec:DBmeasure}, so that $A_f$ is Hermitian with $\|A_f\| \leq 1$ 
and $\mathrm{Tr}(\sigma_\beta A_f) = \mathrm{Tr}(\sigma_\beta A)$. For a fixed 
measurement strength $u \in (0, 1)$, the measurement channel is defined by the 
two Kraus operators
\begin{align}
    K_\pm = \sqrt{\tfrac{1}{2}(I \pm uA_f)}.
\end{align}
Since $u\|A_f\| \leq u < 1$, both $\frac{1}{2}(I \pm uA_f)$ are positive 
semidefinite and satisfy $K_+^\dagger K_+ + K_-^\dagger K_- = I$, so 
$\mathcal{M}(\rho) = K_+\rho K_+^\dagger + K_-\rho K_-^\dagger$ is a valid 
trace-preserving channel. The outcome probabilities when the system is in state 
$\rho$ are
\begin{align}
    p_\pm(\rho) = \mathrm{Tr}\!\left(\rho K_\pm^\dagger K_\pm\right) 
    = \tfrac{1}{2}\bigl(1 \pm u\,\mathrm{Tr}(\rho A_f)\bigr),
\end{align}
so $p_+(\rho) - p_-(\rho) = u\,\mathrm{Tr}(\rho A_f)$. Assigning values 
$z = \pm 1$ to outcomes $\pm$, respectively, the estimator 
\begin{equation*}
    v := z/u    
\end{equation*}
satisfies 
$\mathbb{E}_\rho[v] = \mathrm{Tr}(\rho A_f)$. In particular, when $\rho=\sigma_\beta$,
we have $\mathbb{E}_{\sigma_\beta}[v] = \mathrm{Tr}(\sigma_\beta A_f)
= \mathrm{Tr}(\sigma_\beta A)$. The variance is bounded by
$u^{-2} = \mathcal{O}(1)$ for $u = \Theta(1)$. The block-encoding 
implementations of $K_\pm$ are discussed in 
Sections~\ref{app:implementation-Af} and~\ref{app:implementation-sqrtAf}, 
and the approximate implementation of the full POVM is given in 
Section~\ref{app:implementation-POVM}.

\subsection{The Warm-Start Property}\label{sec:warm-start-proof}
Upon observing outcome $\pm$, the system collapses to the post-selected state
$$\rho_\pm' = \frac{K_\pm \rho K_\pm^\dagger}{p_\pm(\rho)},$$
which is distinct from the post-measurement state $\mathcal{M}[\rho]$ averaged over outcomes.  We show that for $\tau = \Omega(\beta)$, the post-selected state $\rho_\pm'$ 
inherits the warm-start property from $\rho$: if $\chi^2(\rho, \sigma_\beta) = 
\mathcal{O}(1)$, then $\chi^2(\rho_\pm', \sigma_\beta) = \mathcal{O}(1)$ as well.

\begin{theorem}\label{thm:warmstart}
($\chi^2$-warm-start property). Assume $\|A\|\le 1$. Let $\chi^2(\rho, \sigma_\beta) \le C$ for a given $C > 0$ and define $c_f := \exp\!\left(\beta^2/(32\tau^2)\right)$. For $\tau = \Omega(\beta)$ and the measurement strength $u \in (0, 1/c_f)$, the post-selected state $\rho_{\pm}'$ satisfies:
$$\chi^2(\rho_{\pm}',\sigma_\beta) \le \frac{1+C}{(1-u)^2(1-c_f u)^2}.$$
\end{theorem}

For a sufficiently small constant $u = \Theta(1)$ and $\tau = \Omega(\beta)$, the right-hand side is 
$\mathcal{O}(1 + C)$. In particular, if $\chi^2(\rho, \sigma_\beta) = \mathcal{O}(1)$ 
after the burn-in period, then the post-selected state satisfies 
$\chi^2(\rho_\pm', \sigma_\beta) = \mathcal{O}(1)$ as well. This decouples the 
per-sample remixing cost from $\sigma_{\beta,\mathrm{min}}$: under a spectral gap 
$\lambda > 0$, the post-selected state returns to trace-norm $\epsilon'$-proximity of 
$\sigma_\beta$ in time $\mathcal{O}(\lambda^{-1}\log(1/\epsilon'))$, independently 
of $\sigma_{\beta,\mathrm{min}}$.

\begin{proof}[Proof of Theorem~\ref{thm:warmstart}.]
We first establish that multiplication by $A_f$ is bounded in the $\|\cdot\|_{\sigma_\beta,-1/2}$ norm.  For any operator $X$,
\[
\|A_f X\|_{\sigma_\beta,-1/2}
\le \|\sigma_\beta^{-1/4} A_f \sigma_\beta^{1/4}\|
   \|X\|_{\sigma_\beta,-1/2},
\]
and similarly
\[
\|X A_f\|_{\sigma_\beta,-1/2}
\le \|\sigma_\beta^{1/4} A_f \sigma_\beta^{-1/4}\|
   \|X\|_{\sigma_\beta,-1/2}.
\]
By Lemma~\ref{lem:norm-imaginary}, the imaginary-time-shifted norms satisfy
\[
\|e^{\pm\beta H/4} A_f e^{\mp\beta H/4}\|
\le \|A\|\int_{-\infty}^{\infty}|f(t\pm i\beta/4)|\,dt
= \|A\|\,c_f \le c_f,
\]
so $\|A_f X\|_{\sigma_\beta,-1/2} \le c_f\|X\|_{\sigma_\beta,-1/2}$ and likewise $\|XA_f\|_{\sigma_\beta,-1/2} \le c_f\|X\|_{\sigma_\beta,-1/2}$.

Let $B_\pm = \sqrt{I \pm uA_f}$, so that $K_\pm = B_\pm/\sqrt{2}$ and $\rho_\pm' = B_\pm\rho B_\pm\,/\,(2p_\pm(\rho))$.  Since $p_\pm(\rho) = \tfrac{1}{2}(1 \pm u\,\mathrm{Tr}(\rho A_f)) \ge (1-u)/2$, we have $2p_\pm(\rho) \ge 1-u$.  For $c_f u < 1$, the generalized binomial series
$$B_\pm = \sum_{m=0}^{\infty}\binom{1/2}{m}(\pm uA_f)^m$$
converges absolutely.  Applying the multiplication bound iteratively gives $\|A_f^m X A_f^n\|_{\sigma_\beta,-1/2} \le c_f^{m+n}\|X\|_{\sigma_\beta,-1/2}$, and hence we have 
$$\|B_\pm\rho B_\pm\|_{\sigma_\beta,-1/2} \le \left(\sum_{m=0}^{\infty}\left|\binom{1/2}{m}\right|(c_f u)^m\right)^2\|\rho\|_{\sigma_\beta,-1/2} \le \frac{\|\rho\|_{\sigma_\beta,-1/2}}{1-c_f u},$$
where the last inequality uses $\sum_m|\binom{1/2}{m}|x^m \le (1-x)^{-1/2}$ for $x \in [0,1)$.  Since $p_\pm(\rho) \ge (1-u)/2$, we obtain 
$$\|\rho_\pm'\|_{\sigma_\beta,-1/2} \le \frac{\|\rho\|_{\sigma_\beta,-1/2}}{(1-u)(1-c_f u)}.$$
Since $\chi^2(\rho,\sigma_\beta) = \|\rho\|_{\sigma_\beta,-1/2}^2 - 1 \le C$, we have $\|\rho\|_{\sigma_\beta,-1/2} \le \sqrt{1+C}$.  Using $\chi^2(\rho_\pm',\sigma_\beta) \le \|\rho_\pm'\|_{\sigma_\beta,-1/2}^2$ completes the proof. 
\end{proof}

\subsection{Overall Complexity}\label{sec:remixing-complexity}

With the measurement channel $\mathcal{M}$ and the warm-start guarantee of Theorem~\ref{thm:warmstart} in hand, we now describe the full estimation protocol and analyze its complexity. Since $\mathcal{M}$ does not satisfy detailed balance, each measurement is followed by a remixing phase of $k_0$ Gibbs sampling steps to return the state near $\sigma_\beta$. We analyze the sample complexity in the ideal case using martingale concentration, exploiting the boundedness of the outcomes $v_{a_t}$ to obtain $\log(1/\eta)$ dependence on the failure probability. We then perform a stability analysis to account for approximate implementation errors.

\medskip \noindent
\textbf{Protocol.} Given a Gibbs sampling channel $\mathcal{N}$, the measurement channel $\mathcal{M}$ from Section~\ref{sec:POVM-scheme}, an initial state $\rho$, and integers $T_{\mathrm{burn}}$, $k_0$, and $T$, the protocol proceeds as follows:
\begin{itemize}
    \item \textbf{Burn-in stage:} Apply $\mathcal{N}$ for $T_{\mathrm{burn}}:=t_{\mathrm{mix}}(\eta/3)$ steps to obtain $\rho_{\mathrm{burn}}$ satisfying $\|\rho_{\mathrm{burn}} - \sigma_\beta\|_1 \le \eta/3$. Set $\rho_1 = \rho_{\mathrm{burn}}$.
     \item \textbf{Sampling stage:} For $t = 1, 2, \dots, T$, apply $\mathcal{M}$ to $\rho_t$ to obtain outcome label $a_{t} \in \{+,-\}$ and post-selected state
     \[
     \rho_{t,a_{t}}' =
     \frac{K_{a_t}\rho_t K_{a_t}^\dagger}
          {\mathrm{Tr}(\rho_t K_{a_t}^\dagger K_{a_t})}.
     \]
     Then apply $\mathcal{N}^{k_0}$ to $\rho_{t,a_{t}}'$ for
     $k_0= \mathcal{O}\!\left(\frac{1}{\lambda}\log\frac{1}{\epsilon}\right)$ to obtain
     $\rho_{t+1} = \mathcal{N}^{k_0} \left(\rho_{t,a_{t}}'\right)$.

    \item \textbf{Estimation:} Assign real-valued outcomes $v_+ = 1/u$ and $v_- = -1/u$ to the branch labels. Output $\frac{1}{T}\sum_{t=1}^T v_{a_t}$ as the estimator of $\mathrm{Tr}(\sigma_\beta A)$.

\end{itemize}

\medskip \noindent
\textbf{Sample complexity via martingale concentration.} We first analyze the ideal case where $\mathcal{M}$ and $\mathcal{N}$ are implemented exactly, and the initial state is exactly the Gibbs state, i.e. $\rho_1 = \sigma_\beta$.   Choose $k_0 = \mathcal{O}\!\left(\frac{1}{\lambda}\log\frac{1}{\epsilon}\right)$ so that, by Theorem~\ref{thm:warmstart} and the exponential decay of the $\chi^2$-divergence under $\mathcal{N}$, each remixed state satisfies $\|\rho_t - \sigma_\beta\|_1 \le \epsilon/3$ for all $t \ge 1$.  Let $\mathcal{F}_t$ denote the $\sigma$-algebra generated by the outcome labels $(a_1, \dots, a_t)$ up to step $t$.  The outcome $v_{a_t} \in \{-1/u, +1/u\}$ is $\mathcal{F}_t$-measurable with $|v_{a_t}| \le 1/u = \mathcal{O}(1)$, and hence $|v_{a_t}-\mathbb{E}[v_{a_t}\mid\mathcal{F}_{t-1}]|\le 2/u$.  Since $\rho_t$ is determined by $\mathcal{F}_{t-1}$, the conditional mean is $\mathbb{E}[v_{a_t}\mid\mathcal{F}_{t-1}]=\mathrm{Tr}(\rho_t A_f)$. Therefore,
\[
\left|\mathbb{E}[v_{a_t} \mid \mathcal{F}_{t-1}] - \mathrm{Tr}(\sigma_\beta A)\right|
= \left|\mathrm{Tr}\bigl((\rho_t-\sigma_\beta)A_f\bigr)\right|
\le \|\rho_t - \sigma_\beta\|_1
\le \frac{\epsilon}{3}.
\]
We will use the following concentration result from Azuma–Hoeffding inequality \cite[Theorem 7.2.1]{alon2016probabilistic} applied to the martingale differences $X_t-\mathbb{E}[X_t\mid\mathcal{F}_{t-1}]$.

\begin{lemma}\label{lem:concentration}
Let $(X_t)_{t=1}^N$ be real-valued random variables adapted to a filtration $(\mathcal{F}_t)$, where $\mathcal{F}_t$ is generated by $(X_1, \dots, X_t)$.  Suppose $|\mathbb{E}[X_t \mid \mathcal{F}_{t-1}] - \mu| \le \frac{\epsilon}{3}$ and $|X_t - \mathbb{E}[X_t \mid \mathcal{F}_{t-1}]| \le c$ almost surely for all $t$.  Then for $N \ge \frac{18c^2}{\epsilon^2}\log\frac{6}{\eta}$:
$$\mathbb{P}\!\left(\left|\frac{1}{N}\sum_{t=1}^N X_t - \mu\right| \ge \frac{2\epsilon}{3}\right) \le \frac{\eta}{3}.$$
\end{lemma}

Applying Lemma~\ref{lem:concentration} to $(v_{a_t})$ with $\mu = \mathrm{Tr}(\sigma_\beta A)$ and $c = 2/u = \mathcal{O}(1)$, we conclude that $T = \mathcal{O}(\epsilon^{-2}\log(1/\eta))$ applications of the measurement channel $\mathcal{M}$ suffice to estimate $\mathrm{Tr}(\sigma_\beta A)$ to accuracy $\frac{2\epsilon}{3}$ with failure probability at most $\frac{\eta}{3}$ in the ideal case.  The total number of Gibbs sampling steps is:
$$T \cdot k_0 = \mathcal{O}\!\left(\frac{1}{\lambda\,\epsilon^2}\log\frac{1}{\eta}\cdot\log\frac{1}{\epsilon}\right).$$

\medskip \noindent
\textbf{Stability analysis for approximate implementation.}
Each measurement step consists of applying $\mathcal{M}$ followed by $\mathcal{N}^{k_0}$, which defines a quantum instrument $\{\mathcal{E}_a\}_{a\in\{+,-\}}$ with $\mathcal{E}_a(\rho) = \mathcal{N}^{k_0}(K_a\rho K_a^\dagger)$. A trajectory of outcome labels $\vec{a}=(a_1,\dots,a_T)\in\{+,-\}^T$ induces the ideal trajectory distribution:
$$\mathbb{P}(\vec{a}) = \mathrm{Tr}\!\left(\mathcal{E}_{a_T}\circ\cdots\circ\mathcal{E}_{a_1}(\sigma_\beta)\right).$$
In practice, $\mathcal{M}$ and $\mathcal{N}^{k_0}$ are only approximately implementable. We instead apply an $\epsilon'$-approximate instrument $\{\tilde{\mathcal{E}}_a\}$ (in the sense of Definition~\ref{def:approx-instrument}) starting from $\rho_{\mathrm{burn}}$, inducing the approximate trajectory distribution:
$$\tilde{\mathbb{P}}(\vec{a}) = \mathrm{Tr}\!\left(\tilde{\mathcal{E}}_{a_T}\circ\cdots\circ\tilde{\mathcal{E}}_{a_1}(\rho_{\mathrm{burn}})\right).$$
Since $\|\sigma_\beta - \rho_{\mathrm{burn}}\|_1 \le \frac{\eta}{3}$ from the burn-in, Lemma~\ref{lem:error-accum} gives:
$$\|\mathbb{P} - \tilde{\mathbb{P}}\|_{\mathrm{TV}} \le \frac{1}{2}\!\left(\frac{\eta}{3} + T\epsilon'\right).$$
Define the failure event $S := \bigl\{|\tfrac{1}{T}\sum_{t=1}^T v_{a_t} - \mathrm{Tr}(\sigma_\beta A)| \ge \epsilon\bigr\}$. The ideal protocol satisfies $\mathbb{P}(S) \le \frac{\eta}{3}$ by the martingale analysis in Lemma \ref{lem:concentration}. Therefore:
$$\tilde{\mathbb{P}}(S) \le \mathbb{P}(S) + \|\mathbb{P} - \tilde{\mathbb{P}}\|_{\mathrm{TV}} \le \frac{\eta}{3} + \frac{\eta}{6} + \frac{T\epsilon'}{2}.$$
Setting $\frac{T\epsilon'}{2} \le \frac{\eta}{6}$ gives total failure probability at most $\frac{2\eta}{3} \le \eta$. Substituting $T = \mathcal{O}(\epsilon^{-2}\log(1/\eta))$, the required per-step accuracy is:
$$\epsilon' = \mathcal{O}\!\left(\frac{\epsilon^2\,\eta}{\log(1/\eta)}\right).$$

Since the implementation cost of $\mathcal{M}$ (Sections~\ref{app:implementation-Af}--\ref{app:implementation-POVM}) and $\cN$~\cite{chen2025efficientexactnoncommutativequantum,Ding_2025} scales logarithmically with the precision $\epsilon'$, the cost incurred by the approximation introduces only a minor overhead.

\section*{Acknowledgments}
B.L. is supported by National Key R$\&$D Program of China, Grant No. 2024YFA1016000. H.C. and L.Y. are supported by the U.S. Department of Energy, Office of Science, Accelerated Research in Quantum Computing Centers, Quantum Utility through Advanced Computational Quantum Algorithms, Grant No. DESC002557. J.J. is supported by the Simons Quantum Postdoctoral Fellowship and by a Simons Investigator Award in Mathematics through Grant No. 825053.

\appendix

\section{The Measurement Process and the Outcome Statistics}\label{app:measurement-stats}
To provide a unified framework for analyzing the sample complexity (governed by the correlation time) and the stability against numerical implementation errors, we formalize our sequential measurement protocols using the language of quantum instruments.

We describe a general quantum instrument, which yields a classical outcome branch label from a finite set $a \in \{1, 2, \dots, s\}$, by a collection of quantum superoperators $\{\mathcal{E}_a\}_{a=1}^s$ that satisfies:
\begin{itemize}
\item  The sum over all possible measurement branches, denoted as $\mathcal{E} := \sum_{a=1}^s \mathcal{E}_a$, is a valid quantum channel (i.e., a CPTP map)
\item Each individual map $\mathcal{E}_a$ is completely positive (CP) and trace non-increasing.
\end{itemize}
Each branch label $a$ is associated with a real-valued outcome $v_a \in \mathbb{R}$ via a function $v: \{1,\dots,s\} \to \mathbb{R}$. For a quantum system initialized in a density matrix $\rho$, the application of this quantum instrument yields the branch label $a$ with probability:
$$p_a = \Tr(\mathcal{E}_a(\rho))$$
Conditioned on observing branch $a$, the post-selected state is updated to:
$$\rho^{(a)} = \frac{\mathcal{E}_a(\rho)}{\text{Tr}(\mathcal{E}_a(\rho))} = \frac{\mathcal{E}_a(\rho)}{p_a}$$

In our framework for predicting observables, we apply this quantum instrument sequentially, which forms a measurement process. Starting from an initial state $\rho_0$, we apply the instrument at discrete time steps $t = 1, 2, \dots, T$. Let $\vec{a} = (a_1, a_2, \dots, a_T) \in \{1, \dots, s\}^T$ denote a specific sequence, or trajectory, of recorded branch labels over these $T$ steps, with corresponding real-valued outcomes $v_{a_t}$.

The joint probability distribution of observing a specific trajectory $\vec{a}$ is defined by:
$$\mathbb{P}(\vec{a}) = \text{Tr}\left( \mathcal{E}_{a_T} \circ \mathcal{E}_{a_{T-1}} \circ \cdots \circ \mathcal{E}_{a_1}(\rho_0) \right)$$

We can now cleanly map this unified framework onto the two primary measurement strategies developed in this paper.

\medskip \noindent
\textbf{Detailed-Balance Measurement (Section \ref{sec:DBmeasure})} For the detailed-balance construction, the measurement part is decomposed into three branches: $\mathcal{M} = \mathcal{M}_1 + \mathcal{M}_2 + \mathcal{M}_3$, where each branch acts via a single Kraus operator $\mathcal{M}_a(\rho) = K_a \rho K_a^\dagger$, where $K_a$ is defined in Theorem \ref{thm:DB-channel}, with real-valued outcomes $v_1 = v_2 = 1/(2c^2u)$ and $v_3 = 0$. Let $\mathcal{N}$ denote the one-step Gibbs sampling channel. The effective single-step map of the quantum instrument corresponding to observing branch label $a$ is:
$$\mathcal{E}_a = \mathcal{M} \circ \mathcal{N} \circ \mathcal{M}_a$$
Summing over all branches yields the aggregate unconditional channel $\mathcal{E} = \mathcal{M} \circ \mathcal{N} \circ \mathcal{M}$.

\medskip \noindent
\textbf{Measurement-Remixing (Section \ref{sec:warm-start})} For the warm-start strategy, the measurement channel consists of two branches of measurement outcomes: $\mathcal{M} = \mathcal{M}_+ + \mathcal{M}_-$, where $\mathcal{M}_\pm(\rho) = K_\pm \rho K_\pm^\dag$ is defined in Section \ref{sec:POVM-scheme}.  Let $\mathcal{N}$ denote the single-step Gibbs sampling channel and $k_0$ denote a time that suffices for the remixing from warm start. In this case, the effective map for outcome $i$ is 
$$\mathcal{E}_\pm =  \mathcal{N}^{k_0} \circ \mathcal{M}_\pm.$$

\subsection{Stability Analysis of  Implementation Errors}\label{sec:stability}
In practice, executing these operations on a quantum computer introduces implementation errors (e.g., from block-encoding). Instead of the exact measurement branches $\mathcal{M}_a$ and the ideal mixing channel $\mathcal{N}$, we implement approximations. We quantify this numerical accuracy using the following definition:
\begin{definition} \label{def:approx-instrument} We say that a quantum instrument $\{\tilde{\mathcal{E}}_a\}_{a=1}^s$ is an $\epsilon$-approximate implementation of the ideal quantum instrument $\{\mathcal{E}_a\}_{a=1}^s$ if, for any input density matrix $\rho$, the single-step statistical deviation is bounded by:
$$\sum_{a=1}^s \| \tilde{\mathcal{E}}_a(\rho) - \mathcal{E}_a(\rho) \|_{1} \le \epsilon.$$
\end{definition}
This error model translates into a sequence of noisy trajectory maps. Let $\tilde{\mathbb{P}}(\vec{a})$ denote the probability distribution of the trajectory branch labels generated by recursively applying the approximate instrument $\{\tilde{\mathcal{E}}_a\}_{a=1}^s$ to the initial state $\rho_0$. By bounding the distance between the exact composition and its noisy counterpart, we can rigorously bound the total variation distance between the ideal and physically realized trajectory distributions.

\begin{lemma}[Error Accumulation in Sequential Measurements] \label{lem:error-accum} Let $\mathbb{P}$ be the trajectory distribution generated by the ideal quantum instrument $\{\mathcal{E}_a\}_{a=1}^s$ over $T$ steps starting from $\rho_0$, and $\tilde{\mathbb{P}}$ be the distribution generated by an $\epsilon$-approximate implementation $\{\tilde{\mathcal{E}}_a\}_{a=1}^s$ starting from $\tilde{\rho}_0$, with $\|\rho_0 - \tilde{\rho}_0\|_1 \le \eta$. The total variation distance between the two trajectory distributions is bounded by
$$\| \mathbb{P} - \tilde{\mathbb{P}} \|_{\mathrm{TV}} \le \frac{1}{2}(\eta + T \epsilon).$$
\end{lemma}
\begin{proof}

To bound the TV distance between the classical trajectory distributions, we shall embed the classical outcomes into auxiliary quantum registers. To do so, let $\mathcal{H}_S$ denote the Hilbert space of the quantum system. At each time step $t \in \{1, \dots, T\}$, we introduce a fresh classical register $C_t$. Let $C_{<t}$ denote the joint classical register containing all previous measurement outcomes up to step $t-1$. 
We now define the ideal and approximate Quantum-Classical (QC) step-dependent channels $\Phi_t,\tilde{\Phi}_t: \mathcal{B}(\mathcal{H}_{C_{<t}} \otimes \mathcal{H}_S) \to \mathcal{B}(\mathcal{H}_{C_{\le t}} \otimes \mathcal{H}_S)$ as follows \cite[Section~4.4.4]{khatri2020principles}. For $1 \le  t \le T$, define  
\begin{align*}
    \Phi_t(\rho) = \sum_{a=1}^s |a\rangle\langle a|_t \otimes (\id_{< t} \otimes \mathcal{E}_a)(\rho), \quad \tilde{\Phi}_t(\rho) = \sum_{a=1}^s |a\rangle\langle a|_t \otimes  (\id_{< t} \otimes \tilde{\mathcal{E}}_a)(\rho)
\end{align*}
where $\{|a\rangle\}_{a=1}^s$ is an orthonormal basis for the $t$-th classical register, and $\mathrm{id}_{<t}$ denotes the identity map on the previous classical registers $1,\dots,t-1$.

Starting from initial states $\rho_0$ and $\tilde{\rho}_0$ respectively, define recursively
\begin{align*}
\rho_t := \Phi_t\Phi_{t-1}\cdots \Phi_1(\rho_0), \qquad
\widetilde \rho_t := \widetilde\Phi_t\widetilde\Phi_{t-1}\cdots \widetilde\Phi_1(\tilde{\rho}_0).
\end{align*}
Then we have 
$$
\rho_t =  \sum_{\vec{a}_{\leq t}} |\vec{a}_{\leq t}\rangle\langle\vec{a}_{\leq t}| \otimes (\mathcal{E}_{a_t} \circ \dots \circ \mathcal{E}_{a_1})(\rho_0) = \sum_{\vec{a}_{\le t}} \mathbb{P}(\vec{a}_{\le t}) |\vec{a}_{\le t}\rangle\langle\vec{a}_{\le t}| \otimes \rho_{\vec{a}_{\le t}},$$
where $\mathbb{P}(\vec{a}_{\le t}) = \mathrm{Tr}[(\mathcal{E}_{a_{t}} \circ \dots \circ \mathcal{E}_{a_1})(\rho_0)]$ is the probability of the trajectory occurring, and $\rho_{\vec{a}_{\le t}}$ is the normalized post-selected state conditioned on that trajectory.

Note that tracing out the system isolates the purely classical probability distribution over the trajectory: $\mathrm{Tr}_S(\rho_T) = \sum_{\vec{a}} \mathbb{P}(\vec{a})|\vec{a}\rangle \langle \vec{a}|$. Similarly, $\mathrm{Tr}_S(\tilde{\rho}_T) = \sum_{\vec{a}} \tilde{\mathbb{P}}(\vec{a})|\vec{a}\rangle \langle \vec{a}|$. By the monotonicity of the trace distance under the partial trace, the TV distance between the classical trajectory probabilities is bounded by the trace distance of the full joint states: 
\begin{align} \label{eq:TVbound}
    \| \mathbb{P} - \tilde{\mathbb{P}} \|_{\mathrm{TV}} = \frac{1}{2} \| \text{Tr}_{{S}}(\rho_T) - \text{Tr}_{S}(\tilde{\rho}_T) \|_1 \le \frac{1}{2} \|\rho_T-\widetilde\rho_T\|_1 
\end{align}
Therefore, it remains to bound $\|\rho_T - \tilde{\rho}_T\|_1$. 

For this, we first claim that 
\begin{equation}\label{eq:qc-step-bound}
\bigl\|(\Phi_t-\widetilde \Phi_t)(\tilde{\rho}_{t-1})\bigr\|_1 \le \epsilon .
\end{equation}
Indeed, writing $\Delta_i := \mathcal E_i-\widetilde{\mathcal E}_i$, we have
\begin{equation*}
\Phi_t-\widetilde\Phi_t
= \sum_{a=1}^s
|a\rangle\langle a|_t \otimes
(\mathrm{id}_{<t}\otimes \Delta_a).
\end{equation*}
Since the output is block diagonal in the $t$-th classical register, its trace norm decomposes as the sum of the trace norms of the blocks: 
\begin{align*}
\bigl\|(\Phi_t-\widetilde\Phi_t)(\tilde{\rho}_{t-1})\bigr\|_1
&= \sum_{a=1}^s \Big\| \sum_{\vec{a}_{<t}} \tilde{\mathbb{P}}(\vec{a}_{<t})\,
|\vec{a}_{<t}\rangle\langle \vec{a}_{<t}| \otimes \Delta_a(\tilde{\rho}_{\vec{a}_{<t}}) \Big\|_1 \\
&= \sum_{a=1}^s  \sum_{\vec{a}_{<t}} \tilde{\mathbb{P}}(\vec{a}_{<t})
\bigl\|\Delta_a(\tilde{\rho}_{\vec{a}_{<t}})\bigr\|_1 \\
&\le \sum_{\vec{a}_{<t}} \tilde{\mathbb{P}}(\vec{a}_{<t})\,\epsilon = \epsilon\,,
\end{align*}
where in the second line we used that the operator is also block diagonal in the previous classical registers $\ket{\vec{a}_{<t}}$, and in the third line we used the $\epsilon$-approximation assumption for the instrument.

Therefore, using \eqref{eq:qc-step-bound} and contractivity of the trace norm under CPTP maps, we obtain
\begin{align*}
\|\rho_t-\widetilde\rho_t\|_1
&= \|\Phi_t(\rho_{t-1})-\widetilde\Phi_t(\widetilde\rho_{t-1})\|_1 \\
&\le \|\Phi_t(\rho_{t-1}-\widetilde\rho_{t-1})\|_1
+ \|(\Phi_t-\widetilde\Phi_t)(\widetilde\rho_{t-1})\|_1 \\
&\le \|\rho_{t-1}-\widetilde\rho_{t-1}\|_1 + \epsilon.
\end{align*}
Iterating this bound from $t=1$ to $T$ and using $\|\rho_0 - \tilde{\rho}_0\|_1 \le \eta$ as the base case yields
\begin{equation*}
\|\rho_T-\widetilde\rho_T\|_1 \le \eta + T\epsilon.
\end{equation*}
Combining this with \eqref{eq:TVbound}, we complete the proof.
\end{proof}
This stability result is applied to the analysis of the inexact implementation for the protocols in Section~\ref{sec:single-traj-DB} and Section~\ref{sec:remixing-complexity}. The cost of achieving an $\epsilon'$-approximate implementation of the measurement channel $\mathcal{M}$ is analyzed in Appendix~\ref{app:implementation}.
\subsection{The Autocorrelation Time}\label{app:autocorr-time}

We now develop the theory of the autocorrelation time within the above quantum instrument framework. Let $\{\mathcal{E}_a\}_{a=1}^s$ be a quantum instrument with aggregate channel $\mathcal{E} = \sum_a \mathcal{E}_a$, and stationary state $\sigma_\beta$ (i.e., $\mathcal{E}(\sigma_\beta) = \sigma_\beta$). The instrument is equipped with a real-valued outcome function $v: \{1,\dots,s\} \to \mathbb{R}$, $a \mapsto v_a$, which assigns a numerical value to each branch label. All statistical quantities defined below, i.e., the mean, variance, autocorrelation, and autocorrelation time, depend on this choice of $v$.

In the stationary regime, the mean and variance of the outcome $v_{a_t}$ at any step $t$ are:
$$\bar{v} := \mathbb{E}_{\sigma_\beta}(\mathcal{E}) = \sum_a v_a\,\mathrm{Tr}(\mathcal{E}_a(\sigma_\beta)), \qquad \mathbb{V}\mathrm{ar}_{\sigma_\beta}(\mathcal{E}) := \sum_a (v_a - \bar{v})^2\,\mathrm{Tr}(\mathcal{E}_a(\sigma_\beta)).$$
Define the centered instrument map:
$$\widehat{\mathcal{E}}_v := \sum_a (v_a - \bar{v})\,\mathcal{E}_a.$$
The autocorrelation function at lag $t \ge 0$ is:
\begin{equation}
\mathbb{C}\mathrm{orr}_{\sigma_\beta}(\mathcal{E}^t) := \mathrm{Tr}\!\left(\widehat{\mathcal{E}}_v \circ \mathcal{E}^t \circ \widehat{\mathcal{E}}_v(\sigma_\beta)\right).
\end{equation}

\begin{lemma}[Covariance formula]\label{lem:covariance}
For any $t, p > 0$, the covariance between outcomes at steps $t$ and $t+p$ in the stationary state is:
$$\mathbb{E}_{\sigma_\beta}\!\left[(v_{a_t} - \bar{v})(v_{a_{t+p}} - \bar{v})\right] = \mathbb{C}\mathrm{orr}_{\sigma_\beta}(\mathcal{E}^{p-1}).$$
\end{lemma}
\begin{proof}
The joint probability of observing $a_t = a$ and $a_{t+p} = b$, marginalizing over the $p-1$ intermediate steps, is $\mathrm{Tr}(\mathcal{E}_b \circ \mathcal{E}^{p-1} \circ \mathcal{E}_a(\sigma_\beta))$ (using stationarity). Therefore, 
$$\mathbb{E}_{\sigma_\beta}\!\left[(v_{a_t}-\bar{v})(v_{a_{t+p}}-\bar{v})\right] = \sum_{a,b}(v_a-\bar{v})(v_b-\bar{v})\,\mathrm{Tr}\!\left(\mathcal{E}_b \circ \mathcal{E}^{p-1} \circ \mathcal{E}_a(\sigma_\beta)\right) = \mathrm{Tr}\!\left(\widehat{\mathcal{E}}_v \circ \mathcal{E}^{p-1} \circ \widehat{\mathcal{E}}_v(\sigma_\beta)\right),$$
which equals $\mathbb{C}\mathrm{orr}_{\sigma_\beta}(\mathcal{E}^{p-1})$ by definition.
\end{proof}

The integrated (finite-size) autocorrelation time for a trajectory of length $T$ is:
\begin{equation}\label{eq:taut-def}
t_{\mathrm{aut},T} := \frac{1}{2} + \sum_{t=1}^T \left(1 - \frac{t}{T}\right)\frac{\mathbb{C}\mathrm{orr}_{\sigma_\beta}(\mathcal{E}^{t-1})}{\mathbb{V}\mathrm{ar}_{\sigma_\beta}(\mathcal{E})}.
\end{equation}
Using Lemma~\ref{lem:covariance} and the definition~\eqref{eq:taut-def}, one verifies directly that the variance of the empirical mean satisfies:
\begin{align}\label{eq:var-empirical-mean}
\mathbb{V}\mathrm{ar}_{\sigma_\beta}\!\left(\frac{1}{T}\sum_{t=1}^T v_{a_t}\right)
&= \frac{1}{T^2}\sum_{t,s=1}^T \mathbb{E}_{\sigma_\beta}\!\left[(v_{a_t}-\bar{v})(v_{a_s}-\bar{v})\right] \notag\\
&= \frac{1}{T^2}\left[T\,\mathbb{V}\mathrm{ar}_{\sigma_\beta}(\mathcal{E}) + 2\sum_{p=1}^{T-1}(T-p)\,\mathbb{C}\mathrm{orr}_{\sigma_\beta}(\mathcal{E}^{p-1})\right] \notag\\
&= \frac{2\,t_{\mathrm{aut},T}\,\mathbb{V}\mathrm{ar}_{\sigma_\beta}(\mathcal{E})}{T},
\end{align}
where the second equality uses translation invariance of the process in the stationary state together with Lemma~\ref{lem:covariance}.

\begin{lemma}[\cite{jiang2026}]\label{lem:sample-complexity}
Let \(\{\mc{E}_a\}_{a\in\mathcal A}\) be a quantum instrument with real outcomes \(\{v_a\}_{a\in\mathcal A}\), and let
$
\mc{E}:=\sum_{a\in\mathcal A} \mc{E}_a
$
be its nonselective channel. Assume that \(\mc{E}\) is trace-preserving and fixes
\(\sigma_\beta\), i.e. $
\mc{E}(\sigma_\beta)=\sigma_\beta$. Let $\rho$ be an initial state with $\|\rho - \sigma_\beta\|_1 \le \eta$. Define the trajectory distributions
$$\mathbb{P}_\rho(\vec{a}) := \mathrm{Tr}\!\left(\mathcal{E}_{a_T} \circ \cdots \circ \mathcal{E}_{a_1}(\rho)\right), \qquad \mathbb{P}_{\sigma_\beta}(\vec{a}) := \mathrm{Tr}\!\left(\mathcal{E}_{a_T} \circ \cdots \circ \mathcal{E}_{a_1}(\sigma_\beta)\right).$$
For any $\epsilon > 0$, if
$$T \ge \frac{2\,\mathbb{V}\mathrm{ar}_{\sigma_\beta}(\mathcal{E})}{\epsilon^2 \eta}\,t_{\mathrm{aut},T},$$
then
$$\mathbb{P}_\rho\!\left(\left|\frac{1}{T}\sum_{t=1}^T v_{a_t} - \mathbb{E}_{\sigma_\beta}(\mathcal{E})\right| \ge \epsilon\right) \le 2\eta.$$
\end{lemma}
\begin{proof}
Let $S := \{\vec{a} : |\frac{1}{T}\sum_t v_{a_t} - \mathbb{E}_{\sigma_\beta}(\mathcal{E})| \ge \epsilon\}$ be the failure event. Applying Lemma~\ref{lem:error-accum} with both the ideal and the comparison process using the same instrument $\{\mathcal{E}_a\}$ (so the per-step approximation error is $\epsilon' = 0$), starting from $\rho_0 = \rho$ and $\tilde{\rho}_0 = \sigma_\beta$ respectively, with $\|\rho_0 - \tilde{\rho}_0\|_1 = \|\rho - \sigma_\beta\|_1 \le \eta$, we obtain
$$\|\mathbb{P}_\rho - \mathbb{P}_{\sigma_\beta}\|_{\mathrm{TV}} \le \tfrac{1}{2}(\eta + T \cdot 0) = \tfrac{\eta}{2}.$$
By the Chebyshev inequality and~\eqref{eq:var-empirical-mean}, $\mathbb{P}_{\sigma_\beta}(S) \le \eta$. Therefore:
\[\mathbb{P}_\rho(S) \le \mathbb{P}_{\sigma_\beta}(S) + \|\mathbb{P}_\rho - \mathbb{P}_{\sigma_\beta}\|_{\mathrm{TV}} \le \eta + \tfrac{\eta}{2} \le 2\eta. \qedhere\]
\end{proof}

For the specific structure $\mathcal{E}_a = \mathcal{M} \circ \mathcal{N} \circ \mathcal{M}_a$ arising in both of our measurement strategies, the centered instrument map takes the form $\widehat{\mathcal{E}}_v = \mathcal{M} \circ \mathcal{N} \circ \widehat{\mathcal{M}}$, where $\widehat{\mathcal{M}}:= \sum_a (v_a - \bar{v})\mathcal{M}_a$ is the centered measurement map. In this case, the autocorrelation time admits a spectral bound.

\begin{proposition}[\cite{jiang2026}]\label{prop:spectral-gap-bound}
Suppose $\mathcal{E}_a = \mathcal{M} \circ \mathcal{N} \circ \mathcal{M}_a$, where $\mathcal{N}$ and $\mathcal{M}$ satisfy KMS detailed balance with respect to $\sigma_\beta$, and the centered measurement map $\widehat{\mathcal{M}} = \sum_a (v_a - \bar{v})\mathcal{M}_a$ also satisfies KMS detailed balance. If $\mathcal{N}$ has spectrum in $[0,1]$ and spectral gap $\lambda > 0$, then
$$t_{\mathrm{aut},T} \le \frac{\theta}{\lambda} + \frac{1}{2}, \quad \text{where} \quad \theta := \frac{\sum_{a,b}(v_a - \bar{v})(v_b - \bar{v})\,\mathrm{Tr}\!\left(\mathcal{M}_b \circ \mathcal{M}_a(\sigma_\beta)\right)}{\mathbb{V}\mathrm{ar}_{\sigma_\beta}(\mathcal{E})}.$$
\end{proposition}

\section{Implementation of the Measurement Channels} \label{app:implementation}

For completeness, we review the definition of block encodings of non-unitary matrices.

\begin{definition}\label{def:block_encoding_old}
A unitary $U_A$ is said to be an $(\alpha, b, \epsilon)$-block encoding of a matrix $A$ if the top-left block of $U_A$ approximates $A/\alpha$ as follows:
$$\| A - \alpha (\langle 0^b| \otimes I) U_A (|0^b \rangle \otimes I) \| \le \epsilon,$$
where $\alpha$ is called the sub-normalization factor. 
\end{definition}

We collect the key arithmetic properties of block encodings used throughout the paper~\cite{Low_2019, Berryetal2014,camps2023explicitquantumcircuitsblock, Chakraborty19, lin2022lecture}, which allow us to construct complex block encodings from simpler ones. 

\begin{lemma}\label{lem:block}
\begin{itemize}
\item (Product of Block Encodings) Let $A_0, \ldots, A_{M-1}$ be matrices with compatible dimensions. Assume that for each $j \in \{0,1,\ldots,M-1\}$, we are given an $(\alpha_j, b_j, \epsilon_j)$-block encoding $U_{A_j}$. Let $U$ be the unitary obtained by applying $U_{A_0}, U_{A_1}, \ldots, U_{A_{M-1}}$ sequentially on disjoint ancilla registers and a common system register. Then $U$ is a block encoding of $A_{M-1}\cdots A_1A_0$ with parameters
\begin{equation*}
\left(\prod_{j=0}^{M-1} \alpha_j,\; \sum_{j=0}^{M-1} b_j,\; \prod_{j=0}^{M-1}(\alpha_j + \epsilon_j) - \prod_{j=0}^{M-1} \alpha_j\right).
\end{equation*}
In particular, if $U_A$ is an $(\alpha, b, \epsilon)$-block encoding of $A$ and $M\epsilon \le \alpha$, then applying $M$ copies of $U_A$ on disjoint ancilla registers yields an $(\alpha^M, Mb, \mathcal{O}(M\alpha^{M-1}\epsilon))$-block encoding of $A^M$.
\item (Linear Combination of Block Encodings) Let $A_0,\dots,A_{M-1}$ be matrices and let $c_0,\dots,c_{M-1}\in\mathbb C$, where $c_j=|c_j|e^{i\phi_j}$. For each $j\in\{0,\dots,M-1\}$, suppose $U_{A_j}$ is an $(\alpha_j,b_j,\epsilon_j)$-block encoding of $A_j$. Let $b:=\max_j b_j$ and regard each $U_{A_j}$ as acting on a common $b$-qubit ancilla register by padding with unused ancillas when necessary. Let $m:=\lceil \log_2 M\rceil$, and define the prepare oracle $U_{\mathrm{prep}}$ on the $m$-qubit control register by 
\begin{align*}
U_{\mathrm{prep}}|0^m\rangle
= \frac{1}{\sqrt{\gamma}}
\sum_{j=0}^{M-1}\sqrt{|c_j|\alpha_j}\,|j\rangle\,, \q     \gamma:=\sum_{j=0}^{M-1}|c_j|\alpha_j.
\end{align*}
The select oracle for accessing $U_{A_j}$ is defined by 
\begin{align*}
    U_{\mathrm{sel}}
= \sum_{j=0}^{M-1}|j\rangle\langle j|\otimes e^{i\phi_j}U_{A_j} + \sum_{j=M}^{2^m-1}|j\rangle\langle j|\otimes I.    
\end{align*}
Then, the matrix $W:=(U_{\mathrm{prep}}^\dagger\otimes I)\,U_{\mathrm{sel}}\,(U_{\mathrm{prep}}\otimes I)$ is a $(\gamma,\; m+b,\; \sum_{j=0}^{M-1}|c_j|\epsilon_j)$-block encoding of $\sum_{j=0}^{M-1}c_jA_j$.  
\end{itemize}
\end{lemma}

\subsection{Block encoding of \texorpdfstring{$A_f$}{A\_f} and \texorpdfstring{$A_{\tilde f}$}{A\_tilde-f}}
\label{app:implementation-Af}

To implement the detailed-balanced measurement channel of Section~\ref{sec:DBmeasure} and the POVM
scheme of Section~\ref{sec:warm-start}, we require efficient block encodings of the filtered observables
\begin{equation} \label{def:agobser}
    A_g := \int_{\mathbb R} g(t)\, e^{iHt} A e^{-iHt}\, dt,
\qquad g\in\{f,\tilde f\}.
\end{equation}
Here
\begin{equation} \label{def:filtergaussian}
    f(t)=\frac{1}{\sqrt{2\pi}\tau}\exp\!\left(-\frac{t^2}{2\tau^2}\right),
\qquad
\tilde f(t)=f\!\left(t-\frac{i\beta}{2}\right).
\end{equation}
The operator $A_f$ is Hermitian, while $A_{\tilde f}$ is generally non-Hermitian. Both appear in the constructions of Section~\ref{sec:DBmeasure}, and $A_f$ is also the filtered observable used in the warm-start POVM of Section~\ref{sec:warm-start}. 
We first state a quadrature lemma for operator-valued integrals, which follows from the Poisson summation formula. 

\begin{lemma}[Quadrature lemma]
\label{lem:quadrature}
Let $\phi:\mathbb R\to \mathcal B(\mathcal H)$ be a Schwartz-class
operator-valued function, with Fourier transform
\begin{equation*}
    \widehat \phi(\omega)=\int_{\mathbb R} e^{-i\omega t}\phi(t)\,\rd t.
\end{equation*}
For any even integer $M>0$ and step size $\Delta t>0$, let
$I_M:=\{-M/2,\ldots,M/2-1\}$. Then
\begin{equation*}
    \left\|
\Delta t \sum_{k\in I_M}\phi(k\Delta t)-\int_{\mathbb R}\phi(t)\,\rd t
\right\|
\le
\underbrace{\sum_{n\neq 0}\left\|\widehat\phi\!\left(\frac{2\pi n}{\Delta t}\right)\right\|}_{\text{aliasing error}}
+ \underbrace{
\Delta t\sum_{k\notin I_M}\|\phi(k\Delta t)\|}_{\text{truncation error}}. 
\end{equation*}
\end{lemma}

\begin{proof}
Apply the scalar Poisson summation formula
\cite[Theorem 4.2.2]{Pinsky2002IntroductionTF} entrywise, in any fixed
orthonormal basis of $\mathcal H$. Since $\phi$ is Schwartz class, all sums
below converge absolutely and the resulting entrywise identity is an identity
of bounded operators:
\begin{align*}
\Delta t \sum_{k \in \mathbb{Z}} \phi(k \Delta t) = \sum_{n \in \mathbb{Z}} \widehat{\phi}\left(\frac{2\pi n}{\Delta t}\right),  
\end{align*}
and noting 
\begin{equation*}
    \widehat{\phi}(0) = \int_{\mathbb R} \phi(t)\,\rd t,
\end{equation*}
we have 
\begin{align*}
    \Delta t \sum_{k\in I_M} \phi(k\Delta t) - \int_{\mathbb R}\phi(t)\,\rd t
    = \sum_{n \neq 0} \widehat{\phi}\left(\frac{2\pi n}{\Delta t}\right)
    - \Delta t \sum_{k\notin I_M} \phi(k\Delta t),
\end{align*}
where the first term on the right is the aliasing error from the finite grid spacing (the non-zero Fourier modes), and the second term is the truncation error from cutting off the infinite sum at $M/2$. The proof is complete by the triangle inequality. 
\end{proof}

We next collect the
elementary estimates for the two filters $f$ and $\tilde f$. 

\begin{lemma}[Gaussian filter estimates]
\label{lem:gaussian-filter-estimates}
For the two filter functions $g\in\{f,\tilde f\}$ defined above, it holds that 
\begin{enumerate}
\item Time-domain decay:
\begin{equation} \label{est:gauss1} 
    |g(t)| \le \frac{C_g}{\sqrt{2\pi}\tau}e^{-t^2/(2\tau^2)},
\qquad C_f=1, \qquad C_{\tilde f}=e^{\beta^2/(8\tau^2)}.
\end{equation}
\item Fourier transforms:
\begin{align} \label{est:gauss2} 
    \widehat f(\omega)=e^{-\tau^2\omega^2/2},
\qquad \widehat{\tilde f}(\omega) = e^{\beta^2/(8\tau^2)}
\exp\!\left(-\frac{\tau^2}{2}\Bigl(\omega-\frac{\beta}{2\tau^2}\Bigr)^2\right).
\end{align}
\item $L^1$ norms:
\begin{align} \label{est:gauss3} 
    \|f\|_{L^1}=1,
\qquad
\|\tilde f\|_{L^1}=e^{\beta^2/(8\tau^2)}. 
\end{align}
In particular, if $\tau=\Omega(\beta)$, then $\|g\|_{L^1}= \mc{O}(1)$ for both
$g=f$ and $g=\tilde f$.
\end{enumerate}
\end{lemma}

We now derive the
resulting block-encoding costs of $A_f$ and $A_{\tilde f}$. 

\begin{theorem}[Block encoding of the filtered observables]
\label{thm:Afblock}
Let $H$ be a system Hamiltonian, and $A$ be a Hermitian observable with $\|A\|\le 1$. Suppose that $U_A$ is an $(\alpha_A,b_A,\epsilon_A)$-block encoding of $A$. Let $g\in\{f,\tilde f\}$, where $f$ and $\tilde f$ are given in \eqref{def:filtergaussian} with $\tau = \Omega(\beta)$. Then, for any $\epsilon\in(0,1/2)$, one can construct an $(\alpha_g,b_g,\epsilon_g)$-block encoding of
\begin{equation*}
    A_g = \int_{\mathbb R} g(t)\,e^{iHt}Ae^{-iHt}\, \rd t\,,
\end{equation*} 
with normalization factor $\alpha_g$ and number of ancillas $b_g$ satisfying
\begin{equation*}
    \alpha_g = \cO(\alpha_A), \qquad b_g=b_A+ \cO \bigl(\log(\tau\|H\|+2)+\log\log(1/\epsilon)\bigr),
\end{equation*}
and the accuracy $\epsilon_g = \cO(\epsilon_A + \epsilon)$. Moreover, the implementation uses
\begin{align*}
&1 \text{ query to } U_A,\\
&T=\cO\!\left(\tau\sqrt{\log(1/\epsilon)}\right) \text{ controlled Hamiltonian simulation time.}
\end{align*}
\end{theorem}

Here the controlled Hamiltonian simulation time means the maximum evolution time
appearing in the multiplexed controlled evolutions
$\sum_j |j\rangle\langle j|\otimes e^{\pm iHt_j}$. We treat the LCU preparation oracle for the Gaussian weights and the associated
reversible arithmetic as standard primitives; their overheads are polylogarithmic
in $M$ and $1/\epsilon$ for the discretization used below, and are suppressed in
the stated bounds.

\begin{proof}[Proof of Theorem \ref{thm:Afblock}]
We divide the proof into four steps. 

\medskip

\noindent \emph{Step 1: truncation of the integral.} Define the truncated integral
\begin{equation*}
    A_g^{(T)}:=\int_{-T}^{T} g(t)\,e^{iHt}Ae^{-iHt} \,\rd t\,.   
\end{equation*}
Since unitary conjugation preserves the operator norm and $\|A\|\le 1$,
\[
\|A_g-A_g^{(T)}\|
\le
\int_{|t|>T} |g(t)|\,dt .
\]
For $g=f$ and $g=\tilde f$, we have, by \eqref{est:gauss1} in Lemma \ref{lem:gaussian-filter-estimates},
\begin{align*}
    \|A_g-A_g^{(T)}\|
\le
\frac{2C_g}{\sqrt{2\pi}\tau}\int_T^\infty e^{-t^2/(2\tau^2)}\,dt
=\cO\!\left(C_g e^{-T^2/(2\tau^2)}\right). 
\end{align*}
Therefore, choosing
\begin{align} \label{eq:truncationtime}
    T= \mc{O}\!\left(\tau\sqrt{\log(C_g/\epsilon)}\right)
\end{align}
makes the truncation error at most $\epsilon/3$. In the regime $\tau=\Omega(\beta)$,
we have $C_g=\cO(1)$ for both $g=f$ and $g=\tilde f$, and then 
\begin{align*}
    T= \cO\!\left(\tau\sqrt{\log(1/\epsilon)}\right). 
\end{align*}

\medskip

\noindent
\emph{Step 2: discretization by a uniform grid.}
Let
\begin{equation*}
    t_j:=-T+j\Delta t,\qquad j=0,\dots,M-1,
\qquad
\Delta t:=\frac{2T}{M},
\end{equation*}
and define the discretized operator
\begin{equation} \label{eq:discreag}
    \widetilde A_g
:=
\Delta t\sum_{j=0}^{M-1} g(t_j)\,e^{iHt_j}Ae^{-iHt_j}.
\end{equation}
For convenience, we assume $M = 2^m$ for an integer $m > 0$. Since $T=M\Delta t/2$, the grid $\{t_j\}_{j=0}^{M-1}$ coincides with
$\{k\Delta t\}_{k=-M/2}^{M/2-1}$ in Lemma~\ref{lem:quadrature} upon the relabeling $k=j-M/2$. 

We apply the quadrature bound of Lemma~\ref{lem:quadrature} to the operator-valued function
\begin{equation*}
    \phi(t):=g(t)e^{iHt}Ae^{-iHt}.
\end{equation*}
Since $g$ is a Gaussian or a complex-shifted Gaussian and $\|H\|<\infty$,
this $\phi$ is Schwartz class in operator norm.
To bound the aliasing term in operator norm, we use analyticity rather than an
entrywise Bohr-frequency estimate. For $y>0$ and $\omega\neq 0$, shift the
contour in the definition of $\widehat\phi(\omega)$ to
$t-i\,\mathrm{sgn}(\omega)y$. Since the two Gaussian filters are entire,
Cauchy's theorem gives
\begin{equation*}
\widehat\phi(\omega)
=e^{-|\omega|y}\int_{\mathbb R} e^{-i\omega t}
g(t-i\,\mathrm{sgn}(\omega)y)
e^{iH(t-i\,\mathrm{sgn}(\omega)y)}
A
e^{-iH(t-i\,\mathrm{sgn}(\omega)y)}
\,\rd t .
\end{equation*}
Writing $\sigma=\mathrm{sgn}(\omega)$, we have
$$\|e^{\sigma yH}Ae^{-\sigma yH}\|\le e^{2y\|H\|}\|A\|.$$ Moreover, the shifted
Gaussian integrals satisfy 
\begin{align*}
    \int_{\mathbb R}|f(t-i\sigma y)|\,\rd t
    &= e^{y^2/(2\tau^2)},\\
    \int_{\mathbb R}|\tilde f(t-i\sigma y)|\,\rd t
    &= \exp\!\left(\frac{(\beta/2+\sigma y)^2}{2\tau^2}\right)
    \lesssim C_{\tilde f},
\end{align*}
for $y = \mc{O}(\tau)$ and $\tau=\Omega(\beta)$. Hence, for both $g=f$ and $g=\tilde f$,
\begin{equation*}
    \|\widehat\phi(\omega)\|
    \lesssim C_g\,e^{2y\|H\|}e^{-|\omega|y}.
\end{equation*}
Choosing
$y=\Theta((\|H\|+\tau^{-1})^{-1})$ therefore gives
\begin{equation*}
\|\widehat\phi(\omega)\|\lesssim C_g e^{-y|\omega|}.
\end{equation*}
Applying Lemma~\ref{lem:quadrature}, the aliasing term is bounded by
\begin{equation*}
 \sum_{n\neq 0}\left\|\widehat\phi\!\left(\frac{2\pi n}{\Delta t}\right)\right\|
\lesssim C_g\sum_{n\ge 1}e^{-2\pi n y/\Delta t}.
\end{equation*}
Thus the choice
\begin{equation*}
    \Delta t^{-1}
    =\Theta\!\left((\|H\|+\tau^{-1})\log(C_g/\epsilon)\right)
\end{equation*}
makes the aliasing error at most $\epsilon/3$. Moreover, the truncation error term in Lemma~\ref{lem:quadrature} is bounded by
\begin{equation*}
    \Delta t\sum_{|j|\ge M/2}\|\phi(j\Delta t)\|
\le
\Delta t\sum_{|j|\ge M/2}|g(j\Delta t)|
=\cO(\epsilon),     
\end{equation*}
if 
\begin{equation*}
T=\cO\!\left(\tau\sqrt{\log(C_g/\epsilon)}\right),
\end{equation*}
which is already ensured by Step 1.
Thus, we have proved 
\begin{equation*}
    \|\widetilde A_g-A_g\|\le \epsilon.     
\end{equation*}

Further, a direct computation gives
\begin{equation*}
    M=\frac{2T}{\Delta t}
= \cO\!\left(
(\tau\|H\|+1)\log^{3/2}(C_g/\epsilon)
\right). 
\end{equation*}
Since $C_g=\cO(1)$ in the regime $\tau=\Omega(\beta)$, this implies
$\log M=\cO(\log(\tau\|H\|+2)+\log\log(1/\epsilon))$.

\medskip

\noindent
\emph{Step 3: LCU block encoding of the discretized operator.} For each grid point $t_j$, define
\begin{equation} \label{def:wj}
    W_j:=(I_{\rm anc}\otimes e^{iHt_j})\,U_A\,(I_{\rm anc}\otimes e^{-iHt_j}).     
\end{equation}
Since conjugation by unitaries preserves block-encoding parameters,
each $W_j$ is an $(\alpha_A,b_A,\epsilon_A)$-block encoding of $e^{iHt_j}Ae^{-iHt_j}$.
We write $c_j:=\Delta t\,g(t_j)$ and recall $\widetilde A_g$ in \eqref{eq:discreag}. 
By linear combinations of unitaries (Lemma~\ref{lem:block}), there is a block encoding of $\widetilde A_g$ with normalization 
\begin{equation*}
    \widetilde \alpha_g :=
\alpha_A\sum_{j=0}^{M-1}|c_j| =
\alpha_A\,\Delta t\sum_{j=0}^{M-1}|g(t_j)|, 
\end{equation*}
ancilla cost $b_A+ \log_2 M$, and implementation error $\widetilde \epsilon_g := \epsilon_A\,\Delta t\sum_{j=0}^{M-1}|g(t_j)|$. 

It remains to estimate $\widetilde \alpha_g$ and $\widetilde \epsilon_g$. It suffices to consider the quadrature estimate of $|g(t)|$. 
Using the same choices of $T$ and $\Delta t$ as above, we can derive 
\begin{equation*}
    \Delta t\sum_{j=0}^{M-1}|g(t_j)| = \|g\|_{L^1} + \cO(\epsilon)\,,     
\end{equation*}
which readily gives, by \eqref{est:gauss3}, 
\begin{equation*}
    \widetilde \alpha_g = \alpha_A(\|g\|_{L^1}+ \cO(\epsilon)) = \alpha_A(C_g + \cO(\epsilon)), \quad \widetilde \epsilon_g = (C_g + \cO(\epsilon))\epsilon_A. 
\end{equation*}
Combining this with the discretization error $\|\widetilde A_g-A_g\|\le \epsilon$ yields an $(\alpha_g,b_g,\epsilon_g)$-block encoding of $A_g$ with desired parameters. 

\medskip 

\noindent
\emph{Step 4: query complexity and controlled evolutions.}
The select oracle takes the form
\begin{equation*}
    U_{\rm sel} = \sum_{j=0}^{M-1}|j\rangle\langle j|\otimes e^{i\phi_j}W_j, \qquad c_j=|c_j|e^{i\phi_j} = |\Delta t\,g(t_j)|e^{i\phi_j}.
\end{equation*}
Substituting the definition \eqref{def:wj} of $W_j$, we obtain the factorization
\begin{equation*}
   U_{\rm sel}
=
\Big(\sum_j |j\rangle\langle j|\otimes I_{\rm anc}\otimes e^{i\phi_j}e^{iHt_j}\Big)
(I\otimes U_A)
\Big(\sum_j |j\rangle\langle j|\otimes I_{\rm anc}\otimes e^{-iHt_j}\Big). 
\end{equation*}
Therefore, the whole construction uses exactly one query to $U_A$.
The required controlled Hamiltonian evolutions only involve times $t_j\in[-T,T]$, and hence the maximum evolution time is of the same order as the truncation time $T$ in \eqref{eq:truncationtime}. 
Equivalently, the stated controlled simulation time is the maximum time in the
multiplexed controlled evolutions as above. 
This proves the theorem.
\end{proof}

\subsection{Block encoding of \texorpdfstring{$c\sqrt{I+uA_f}$}{c sqrt(I+uA\_f)} and \texorpdfstring{$c\sqrt{I+uA_{\tilde f}}$}{c sqrt(I+uA\_tilde-f)}}
\label{app:implementation-sqrtAf}

Having obtained an efficient block encoding of the filtered operator $A_g$ (which represents either
$A_f$ or $A_{\tilde f}$) in Theorem~\ref{thm:Afblock}, our next task is to build block encodings
of the square-root operators
\begin{equation*}
    c\sqrt{I+uA_g}\,,\quad g \in \{f,\tilde{f}\}\,,
\end{equation*}
from a block encoding of $A_g$.

We record the Taylor series of the square root function for later use. The binomial series gives, for complex $z \in \mathbb{C}$ with $|z|<1$, 
\begin{equation} \label{eq:basicseries}
    \sqrt{1+z}=\sum_{k=0}^{\infty}\binom{1/2}{k}z^k,
\end{equation}
where the coefficients satisfy $\bigl|\binom{1/2}{k}\bigr|\le 1$ for all $k\ge 0$. The absolute convergence of the series \eqref{eq:basicseries} is guaranteed by 
\begin{equation} \label{eq:Taylor-bound}
    \sum_{k=0}^{\infty}\left|\binom{1/2}{k}\right| r^k
=2-\sqrt{1-r},\qquad r\in[0,1). 
\end{equation}
Truncating \eqref{eq:basicseries} at order $K$ yields the polynomial
\[
P_K(z):=\sum_{k=0}^{K}\binom{1/2}{k}z^k,
\]
with error
\begin{equation} \label{eq:Taylor-truncation}
    \bigl|\sqrt{1+z}-P_K(z)\bigr|
\le \sum_{k=K+1}^{\infty}|z|^k
=\frac{|z|^{K+1}}{1-|z|},\qquad |z|<1.
\end{equation}
Hence, for $|z|\le r<1$, choosing 
\begin{equation*}
    K= \cO\!\left(\frac{\log(1/\epsilon)}{\log(1/r)}\right)
\end{equation*}
ensures the truncation error is $\cO(\epsilon)$.

\subsubsection{The QSVT construction for real filters}
\label{app:QSVT}

When $g = f$ is real, the filtered observable $A_f$ is Hermitian, so one can apply QSVT~\cite{Gily_n_2019} to implement $c\sqrt{I + uA_f}$. In particular, given a block encoding of $A_f$, this requires only a constant number of 
additional ancilla qubits. 

\begin{lemma}
\label{lem:qsvt-sqrt}
Suppose $U_B$ is an $(\alpha,b,\epsilon_0)$-block encoding of a Hermitian matrix $B$.
Let
\begin{equation} \label{eq:range}
    r:=|u|\alpha\le \frac{1}{2}, \qquad
 c\in (0, \frac{1}{4}].
\end{equation}
Then, for any given $\epsilon'\in(0, \tfrac{1}{8}]$, there exists a $(1,b+2,\epsilon)$-block encoding of $c\sqrt{I+uB}$ with
\begin{equation*}
    \epsilon = 4d\sqrt{\epsilon_0/\alpha}+2\epsilon',
\end{equation*}
using
\begin{align*}
&d= \cO(\log(1/\epsilon')) \text{ queries to } U_B \text{ and } U_B^\dagger,\\
&\text{one query to controlled-}U_B, \text{ and}\\
&\cO((b+1)d) \text{ additional one- and two-qubit gates.}
\end{align*}
\end{lemma}

\begin{proof}
We apply Theorem~56 and Corollary~66 of~\cite{Gily_n_2019}. By definition, $U_B$ encodes
$\widetilde B:=B/\alpha$, whose eigenvalues lie in $[-1,1]$.
We apply the function 
\begin{equation*}
    f(x)=c\sqrt{1+u\alpha x}
\end{equation*}
to $\widetilde B$, so that $f(\widetilde B)=c\sqrt{I+uB}$.

We construct a polynomial approximation of $f$ on $[-1,1]$ via \cite[Corollary~66]{Gily_n_2019}, which states: if $f(x)=\sum_{\ell=0}^{\infty} a_\ell x^\ell$ converges for $|x|\le 1+\delta$ and
$\sum_{\ell=0}^{\infty}|a_\ell|(1+\delta)^\ell \le M$, then $f$ can be $\epsilon'$-approximated on $[-1,1]$ by a polynomial $P_d$ of degree
\begin{equation*}
    d= \cO\!\left(\delta^{-1}\log(M/\epsilon')\right)
\end{equation*}
with $\|P_d\|_{L^\infty([-1,1])}\le M$. The Taylor coefficients of $f$ are
\[
a_\ell = c\binom{1/2}{\ell}(u\alpha)^\ell,
\]
and the radius of convergence is $1/r>2$ since $r = |u| \alpha \le 1/2$ by \eqref{eq:range}. Hence, we can take $\delta=1$. Using \eqref{eq:Taylor-bound}, we find 
\begin{equation*}
    M=\sum_{\ell=0}^{\infty}|a_\ell|2^\ell
= c\sum_{\ell=0}^{\infty}\left|\binom{1/2}{\ell}\right|(2r)^\ell
= c\bigl(2-\sqrt{1-2r}\bigr)
\le 2c\le \frac12. 
\end{equation*}
Therefore, there exists a real polynomial $P_d$ of degree
$d = \cO(\log(1/\epsilon'))$ such that
\[
\|f-P_d\|_{L^\infty([-1,1])}\le \epsilon',
\qquad
\|P_d\|_{L^\infty([-1,1])}\le \frac{1}{2}.
\]

By \cite[Theorem 56]{Gily_n_2019}, applying QSVT to the $(\alpha,b,\epsilon_0)$-block encoding $U_B$ yields a
$(1,b+2,\epsilon'')$-block encoding of $P_d(\widetilde B)$ with $\epsilon'' = 4d\sqrt{\epsilon_0/\alpha}+\epsilon'$ using
\begin{align*}
&d \text{ queries to } U_B \text{ and } U_B^\dagger,\\
&\text{one query to controlled-}U_B, \text{ and}\\
&\cO((b+1)d) \text{ additional gates.}
\end{align*}
Finally, since $\|f-P_d\|_{L^\infty([-1,1])}\le \epsilon'$, the triangle inequality gives
$\epsilon = \epsilon''+\epsilon'
= 4d\sqrt{\epsilon_0/\alpha}+2\epsilon'$. 
\end{proof}

\subsubsection{The LCU construction for general filters}
When the filtered operator is non-Hermitian (for example, $A_{\tilde f}$), QSVT is not directly
applicable for constructing a block encoding of the associated matrix function. Instead, we use a
Taylor-series expansion together with LCU to implement the block encoding of $\sqrt{I + uB}$ for a general operator $B$. Alternative quantum eigenvalue-transformation methods may also be applicable in this setting; see, for example, \cite{Takahira_2020,Low_2024,an2024laplacetransformbasedquantum}.

\begin{theorem}
\label{thm:sqrt-be-lcu}
Let $B$ be an operator with $\|B\|\le \alpha$ and an $(\alpha,b,\delta)$-block encoding $U_B$.
Let $u\in\mathbb C$ and $0<c\le 1$ satisfy 
\begin{equation} \label{eq:pararange}
    r:=|u|\alpha \le \frac12,
\qquad
|u|(\alpha+\delta)<1. 
\end{equation}
Then, for any given $\epsilon>0$, one can construct a
$(\gamma,b',\epsilon')$-block encoding of $c\sqrt{I+uB}$ with 
\begin{equation*}
    \gamma = c\bigl(2-\sqrt{1-r}\bigr) = \cO(c),
\qquad b' = \cO\!\left(\frac{b\log(1/\epsilon)}{\log(1/r)}\right),
\end{equation*}
and
\begin{equation} \label{eq:totalerror}
    \epsilon'=\epsilon + \cO(\delta).
\end{equation}
The construction uses
\begin{align*}
K=\cO\!\left(\frac{\log(1/\epsilon)}{\log(1/r)}\right) \text{ queries to } U_B.
\end{align*}
\end{theorem}

\begin{proof}
We expand $c\sqrt{I+uB}
= c\sum_{k=0}^{\infty}\binom{1/2}{k}u^k B^k$
and truncate at order $K$, defining
\[
P_K := c\sum_{k=0}^{K}\binom{1/2}{k}u^k B^k.
\]
By \eqref{eq:Taylor-truncation} and $|u|\|B\|\le |u|\alpha=r$, we have
\[
\|c\sqrt{I+uB}-P_K\|
\le c\,\frac{r^{K+1}}{1-r}.
\]
Thus the choice of $K=\cO\!\left(\frac{\log(1/\epsilon)}{\log(1/r)}\right)$ ensures the truncation error is at most $\epsilon$.

By Lemma~\ref{lem:block} (product of block encodings), the $k$-th power $B^k$ admits an
\[
(\alpha^k,kb,(\alpha+\delta)^k-\alpha^k)
\]
block encoding. Expressing $P_K$ as an LCU and applying Lemma~\ref{lem:block} again yields the normalization constant
\begin{equation*}
\gamma = c\sum_{k=0}^{K}\left|\binom{1/2}{k}\right| (|u|\alpha)^k \le c\sum_{k=0}^{\infty}\left|\binom{1/2}{k}\right| r^k = c\bigl(2-\sqrt{1-r}\bigr)\,,    
\end{equation*}
the ancilla count
\begin{equation*}
    b' = Kb + \lceil \log_2(K+1)\rceil
= \cO\!\left(\frac{b\log(1/\epsilon)}{\log(1/r)}\right),
\end{equation*}
and the block encoding error, by $|u|(\alpha+\delta)<1$ and \eqref{eq:Taylor-bound}, 
\begin{align*}
&c\sum_{k=0}^{K}\left|\binom{1/2}{k}\right| |u|^k\bigl((\alpha+\delta)^k-\alpha^k\bigr) \\
&\le c\Bigl(\bigl(2-\sqrt{1-r-|u|\delta}\bigr)-\bigl(2-\sqrt{1-r}\bigr)\Bigr) \\
&= c\bigl(\sqrt{1-r}-\sqrt{1-r-|u|\delta}\bigr)
= \cO\!\left(\frac{c|u|\delta}{\sqrt{1-r}}\right).
\end{align*}
Since $r \leq 1/2$, $c\le 1$, and $|u|\delta < 1 - r$ by \eqref{eq:pararange}, we conclude with $\mathcal{O}(\delta)$.

Finally, the select oracle $U_{\mathrm{sel}}=\sum_{k=0}^{K} |k\rangle\langle k|\otimes U_B^k$ 
is implemented using $K$ independent $b$-qubit ancilla blocks: for each
$i=1,\dots,K$, apply $U_B$ to the system and the $i$-th ancilla block,
controlled on the predicate $k\ge i$. This uses exactly $K$ queries to $U_B$.
Combining the truncation and encoding errors completes the proof of \eqref{eq:totalerror}. 
\end{proof}

\subsection{Approximate Implementation of Quantum Instruments } \label{app:implementation-POVM}

\begin{theorem}[Approximate implementation of quantum instruments]\label{thm:approx-channel-be}
Let $s = \mathcal{O}(1)$, and consider a quantum instrument $\{\mathcal{M}_i\}_{i=0}^{s-1}$ given by the Kraus operators
\begin{equation*}
    \mathcal M_i(\rho)=K_i \rho K_i^\dagger,
\qquad
\sum_{i=0}^{s-1} K_i^\dagger K_i = I.
\end{equation*}
Suppose that for each $i$, an $(\alpha_i, b_i, \epsilon)$-block encoding of $K_i$ is available, denoted by $U_i$, with $\epsilon\in(0,1/2)$. Let 
$\alpha = \max_i \alpha_i$ and $b= \max_i b_i$. Then there exists a quantum circuit that implements an approximation 
$\{\widetilde{\mathcal M}_i\}_{i=0}^{s-1}$ such that, for any density matrix $\rho$,
\begin{equation*}
    \sum_{i = 0}^{s-1} \|\widetilde{\mathcal M}_i(\rho)-\mathcal M_i(\rho)\|_1 = \mc{O}(\epsilon).
\end{equation*}
The circuit uses
\begin{align*}
&\mc{O}(\alpha \log(1/\epsilon)) \text{ queries to the block encodings } U_i, \text{ and}\\
&b+ \mc{O}(1) \text{ ancilla qubits.}
\end{align*}
Moreover, $\{\widetilde{\mathcal M}_i\}_{i=0}^{s-1}$ is a valid quantum instrument. That is, each branch $\widetilde{\mathcal M}_i$ is completely positive and $\sum_i \widetilde{\mathcal M}_i$ is trace-preserving. 
\end{theorem}

\begin{proof} 
Let $a_{\mathrm{anc}}$ denote the common ancilla register of the rescaled block 
encodings $\widehat{U}_i$, and write
\begin{align} \label{eq:ancilaspace}
    |0\rangle_{\mathrm{anc}} := |0^{\maf b}\rangle_{a_{\mathrm{anc}}}, \quad 
    \maf{b} = b + \mathcal{O}(1).
\end{align}
We first introduce the ideal Stinespring isometry for the instrument:
\begin{align} \label{idealiso}
    V\colon |\psi\rangle \longmapsto \sum_{i=0}^{s-1} |i\rangle_{a_1}|0\rangle_{\rm anc} K_i |\psi\rangle.
\end{align}
Here, $a_1$ denotes an $s$-dimensional outcome register with computational basis $\{|i\rangle\}_{i=0}^{s-1}$, while $a_{\rm anc}$ is a workspace ancilla register large enough to accommodate the 
block-encoding and amplification subroutines (see \eqref{eq:ancilaspace}). Measuring $a_1$ in the computational basis 
realizes the instrument $\{\mathcal{M}_i\}_{i=0}^{s-1}$.
Moreover, $\sum_i K_i^\dagger K_i = I$ implies 
$V^\dagger V = I$.

\medskip
\noindent
\emph{Step 1: Common rescaling of the block encodings.}
To facilitate the LCU and amplification steps below, we rescale the block encodings $U_i$ of $K_i$ to a common normalization constant. 

For each $i$, pad $U_i$ with idle 
ancillas so that all block encodings act on the same number $b$ of ancilla qubits.
After padding, by a standard normalization-dilution gadget, we convert each
\((\alpha_i,b,\epsilon)\)-block encoding into an
\((\alpha,\mathfrak b,\epsilon)\)-block encoding of \(K_i\), where
\(\alpha=\max_i\alpha_i\). Equivalently, the encoded block is multiplied by
\(\alpha_i/\alpha\), while the remaining amplitude is routed to an orthogonal
failure flag.
Denoting the resulting unitary for the block encoding by 
$\widehat{U}_i$, it satisfies
\begin{align} \label{eq:apprkarus}
  \widetilde{K}_i := \alpha\,(\langle 0^{\maf{b}}| \otimes I)\,
    \widehat{U}_i\,(|0^{\maf{b}}\rangle \otimes I)\,,\quad   \left\| K_i - \widetilde{K}_i  \right\| \leq \epsilon.
\end{align}

\medskip
\noindent
\emph{Step 2: Multiplexing via LCU.} We now construct a block encoding of the Stinespring-type operator with $\widetilde{K}_i$ defined in \eqref{eq:apprkarus}:
\begin{align*}
    Y \colon |\psi\rangle \longmapsto  \sum_{i=0}^{s-1} |i\rangle_{a_1}|0\rangle_{\mathrm{anc}}\widetilde{K}_i |\psi\rangle
\end{align*}
via the individual block encodings $\widehat{U}_i$. 

Define the \emph{prepare} and 
\emph{select} oracles by
\begin{align*}
    U_{\mathrm{prep}}\,|0\rangle_{a_1} := \frac{1}{\sqrt{s}}\sum_{i=0}^{s-1} |i\rangle_{a_1},
    \qquad
    U_{\mathrm{sel}} := \sum_{i=0}^{s-1} |i\rangle\langle i|_{a_1} \otimes \widehat{U}_i,
\end{align*}
and set
\begin{align*}
    W := U_{\mathrm{sel}}\,(U_{\mathrm{prep}} \otimes I).
\end{align*}
For every input state $|\psi\rangle \in \mathcal{H}_S$, a direct computation gives
\begin{align*}
 W\,|0\rangle_{a_1}|0\rangle_{\mathrm{anc}}|\psi\rangle
    = \frac{1}{\alpha\sqrt{s}} \sum_{i=0}^{s-1}    |i\rangle_{a_1}|0\rangle_{\mathrm{anc}}\widetilde{K}_i|\psi\rangle 
    + |\perp(\psi)\rangle,
\end{align*}
where $|\perp(\psi)\rangle$ is orthogonal to 
$\mathrm{span}\{|i\rangle_{a_1}|0\rangle_{\mathrm{anc}} \otimes \mathcal{H}_S : i = 0,\ldots,s-1\}$. Defining 
\begin{align*}
J_{\mathrm{in}}|\psi\rangle := |0\rangle_{a_1}|0\rangle_{\mathrm{anc}}|\psi\rangle, \qquad \Pi_{\mathrm{out}} := \sum_{i=0}^{s-1} |i\rangle\langle i|_{a_1}\otimes |0\rangle\langle 0|_{\mathrm{anc}}\otimes I_S\,,
\end{align*}
we have 
$$\Pi_{\mathrm{out}}\,W\,J_{\mathrm{in}}
= \frac{1}{\alpha\sqrt s}\,Y\,,$$ and hence $W$ is 
a projected unitary encoding of $Y$.

\medskip
\noindent
\emph{Step 3: Oblivious amplitude amplification.} 
Recalling $\sum_i K_i^\dagger K_i=I$ and $\|\widetilde K_i-K_i\|\le \epsilon $, we have 
\begin{align*}
  Y^\dagger Y = \sum_{i=0}^{s-1}  \widetilde K_i^\dagger \widetilde K_i = I + \mc{O}(\epsilon)\,. 
\end{align*}
Hence singular values of $\frac{1}{\alpha\sqrt s}\,Y$ lie in $\frac{1}{\alpha\sqrt s}[1-\mc{O}(\epsilon),\, 1+\mc{O}(\epsilon)]$. Let $U_Y := Y(Y^\dagger Y)^{-1/2}$ be the polar isometry of $Y$, which satisfies $U_Y^\dagger U_Y=I$ and $\|U_Y-Y\| = \mc{O}(\epsilon)$. Moreover, note that 
\begin{align*}
    \|Y-V\|^2 =
\Bigl\|\sum_{i=0}^{s-1}  (\widetilde K_i-K_i)^\dagger(\widetilde K_i-K_i)\Bigr\|\le s \epsilon^2\,,
\end{align*}
implying 
\begin{equation} \label{bound:uni}
    \|U_Y-V\| = \mc{O}(\epsilon). 
\end{equation}

Now apply robust oblivious amplitude amplification to the
projected block $\Pi_{\mathrm{out}}W  J_{\mathrm{in}} = \frac{1}{\alpha\sqrt{s}} Y$. A standard QSVT theorem yields a unitary $\widetilde V$, using $\mc{O}(\alpha\log(1/\epsilon))$ queries to $W$ and $W^\dagger$, such that $\|\widetilde V J_{\mathrm{in}} - U_Y\| = \mc{O}(\epsilon)$. Combining this with \eqref{bound:uni}  gives 
\begin{equation} \label{eq:approxV}
    \|\widetilde V J_{\mathrm{in}} - V\| = \mc{O}(\epsilon)\,,
\end{equation}
for the ideal Stinespring isometry \eqref{idealiso}. 

\medskip
\noindent
\emph{Step 4: Instrument implementation circuits.}
We define the implemented branch maps by measuring the index register $a_1$ and tracing out $a_{\rm anc}$:
\begin{equation*}
    \widetilde{\mathcal M}_i(\rho)
:=
\operatorname{Tr}_{\mathrm{anc}}
\Bigl[
\bigl(\langle i|_{a_1}\otimes I\bigr)
\widetilde V
(|0,0\rangle\langle 0,0|_{a_1, {\rm anc}}\otimes \rho)
\widetilde V^\dagger
\bigl(|i\rangle_{a_1}\otimes I\bigr)
\Bigr].    
\end{equation*}
Each $\widetilde{\mathcal M}_i$ is completely positive. 
Moreover, it holds that 
\begin{align*}
    \sum_{i=0}^{s-1}  \widetilde{\mathcal M}_i(\rho) = \operatorname{Tr}_{a_1,\mathrm{anc}} \Bigl[ \widetilde V (|0,0\rangle\langle 0,0|_{a_1, {\rm anc}}\otimes \rho) \widetilde V^\dagger \Bigr].
\end{align*}
Since $\widetilde{V}$ is unitary, the map $\sum_i \widetilde{\mathcal M}_i$ is also trace-preserving. Thus $\{\widetilde{\mathcal M}_i\}_{i=0}^{s-1}$ is a valid quantum instrument approximating the ideal one.  Indeed, let
\begin{equation*}
    \rho' := V \rho V^\dagger\,,
\qquad \widetilde \rho' := \widetilde V(|0,0\rangle\langle 0,0|_{a_1, {\rm anc}}\otimes \rho)\widetilde V^\dagger.
\end{equation*}
From \eqref{eq:approxV}, we find 
\begin{align*}
    \|\widetilde \rho' - \rho'\|_1 = \mc{O}(\epsilon).
\end{align*}
Here we used the elementary estimate
$\|A\rho A^\dagger-B\rho B^\dagger\|_1
\le (\|A\|+\|B\|)\|A-B\|$ for density matrices $\rho$.
Let $\mathcal D$ be the completely positive trace-preserving map that dephases
the outcome register in the basis $\{|i\rangle\}_{i=0}^{s-1}$ and then traces
out $a_{\rm anc}$. By contractivity of trace distance under $\mathcal D$,
\begin{equation*}
    \left\|
    \sum_{i=0}^{s-1}|i\rangle\langle i|\otimes
    \bigl(\widetilde{\mathcal M}_i(\rho)-\mathcal M_i(\rho)\bigr)
    \right\|_1
    \le
    \|\widetilde \rho' - \rho'\|_1 .
\end{equation*}
Since the operator on the left is block diagonal in the classical outcome
register, its trace norm is the sum of the trace norms of the blocks. Therefore,
\begin{align*}
    \sum_{i=0}^{s-1} \|\widetilde{\mathcal M}_i(\rho)-\mathcal M_i(\rho)\|_1 \le \|\widetilde \rho' - \rho'\|_1 = \mc{O}(\epsilon)\,.
\end{align*}
This proves the theorem.
\end{proof}

\section{Implementation of the superoperator \texorpdfstring{$\mc{R}$}{R}}
\label{sec:implementR}

To implement the rejection branch $\mathcal{R}[\cdot] = K\,(\cdot)\,K^\dagger$ 
with the Kraus operator
\begin{equation} \label{eq:operatoro}
    K = \sqrt{\sqrt{\sigma_\beta}(I - O)\sqrt{\sigma_\beta}}\;\sigma_\beta^{-1/2},
    \qquad
    O := \mathcal{T}'^\dagger[I] = K_1^\dagger K_1 + K_2^\dagger K_2,
\end{equation}
we employ the implementation techniques developed in~\cite[Appendix~C]{gilyen2026} 
for discrete-time quantum Gibbs samplers.

The key structural condition required is that $O$ be \emph{quasi-local in 
energy} with respect to $H$. To make this precise, recall that any operator 
$X \in \mathcal{B}(\mathcal{H})$ admits a Bohr-frequency decomposition
\begin{equation*}
    X = \sum_\nu X_\nu, \qquad e^{iHt} X_\nu e^{-iHt} = e^{i\nu t} X_\nu,
\end{equation*}
where $\nu = E_k - E_j$ ranges over the Bohr frequencies of $H = \sum_k 
E_k|k\rangle\langle k|$.

\begin{definition}[Quasi-local in energy]\label{def:quasilocal}
An operator $X \in \mathcal{B}(\mathcal{H})$ is $(\Lambda, \epsilon')$-quasi-local 
in energy with respect to $H$ if there exists an operator $\widetilde{X}$ such 
that $\|X - \widetilde{X}\| \leq \epsilon'$ and $\widetilde{X}_\nu = 0$ for 
all $|\nu| \geq \Lambda$, or equivalently,
\begin{equation*}
    \langle E_j | \widetilde{X} | E_k \rangle = 0 \qquad 
\text{whenever } |E_j - E_k| \geq \Lambda. 
\end{equation*}
\end{definition}

We use the following results from \cite[Proposition~5 and Lemma~12]{gilyen2026}.

\begin{proposition}
\label{thm:implementR}
Let $\epsilon \in (0, 1/4]$, $s = \Theta(\log(\beta\|H\|))$, and 
$\alpha \geq \|H\|$. Suppose $\|O\| \leq \frac{10^{-3}}{s^2}$ and $O$ is 
$(\Lambda, \epsilon')$-quasi-local in energy with respect to $H$, where
\begin{align} \label{eq:paraomegapeo}
    \Lambda = \frac{1}{10\lfloor\log_2(1/\epsilon)\rfloor}, \qquad 
\epsilon' = \mathcal{O}\!\left(\frac{\epsilon}{s^2\,\mathrm{polylog}(1/\epsilon)}\right).
\end{align}
Then one can implement an $\epsilon$-approximate block-encoding of
\begin{equation*}
    \frac12\sqrt{\sqrt{\sigma_\beta}(I-O)\sqrt{\sigma_\beta}}\;\sigma_\beta^{-1/2}
\end{equation*}
using
\begin{align*}
&\cO(\alpha\beta\,\mathrm{polylog}(\alpha\beta/\epsilon)) \text{ queries to a block encoding of } H/\alpha,\\
&\cO(\mathrm{polylog}(1/\epsilon)) \text{ queries to a block encoding of } O, \text{ and}\\
&\cO(\mathrm{polylog}(\alpha\beta/\epsilon)) \text{ ancillary qubits.}
\end{align*}
\end{proposition}

In our setting, the small-norm condition on $O$ in~\eqref{eq:operatoro} is 
satisfied by the choice $$c = \mathcal{O}(1/\log(\beta\|H\|)),$$
in our constructions \eqref{def:k1} and \eqref{def:k23}, since $O = \mathcal{T}'^\dagger[I] \leq 3c^2 I$ by~\eqref{eq:boundtio}. It therefore 
remains to verify the quasi-locality in energy condition. 

\begin{lemma}[Truncated Gaussian filtered operators]
\label{lem:gaussian-bandlimit}
Let $\chi \in C_c^\infty(\mathbb{R})$ be a smooth cutoff function satisfying
\begin{align*}
    0 \leq \chi \leq 1, \qquad
\chi(x) = 1 \ \text{on} \ |x| \leq \tfrac{1}{2}, \qquad
\operatorname{supp}\chi \subset (-1, 1). 
\end{align*}
For $\Lambda_0 > 0$, define $\chi_{\Lambda_0}(\nu) := \chi\!\left(\nu/\Lambda_0\right)$.
For $g \in \{f, \tilde{f}\}$, define the truncated operator by
\begin{equation} \label{def:oppag}
    A_g^{(\Lambda_0)} := \sum_\nu \chi_{\Lambda_0}(\nu)\,\widehat{g}(-\nu)\,A_\nu.
\end{equation}
Then $A_g^{(\Lambda_0)}$ has Bohr-frequency support contained in 
$[-\Lambda_0, \Lambda_0]$. 
Moreover, let $\|A\| \leq 1$ and 
$\tau = \Omega(\beta)$. It holds that, for constants $c_g>0$, i.e., $c_f=1/16$ and $c_{\tilde f}=1/32$,
\begin{equation*}
    \|A_g - A_g^{(\Lambda_0)}\|
    = \cO\!\left((1 + \tau\Lambda_0)^3
    \exp(-c_g\tau^2\Lambda_0^2)\right)
\qquad g \in \{f, \tilde{f}\}.
\end{equation*}
\end{lemma}

\begin{proof}
Define
\begin{equation*}
    m_{g,\Lambda_0}(\nu) := \bigl(1 - \chi_{\Lambda_0}(\nu)\bigr)\widehat{g}(-\nu),
\qquad
\eta_{g,\Lambda_0}(t) 
:= \frac{1}{2\pi}\int_{\mathbb{R}} m_{g,\Lambda_0}(\nu)\,e^{-i\nu t}\,\mathrm{d}\nu. 
\end{equation*}
Since $A_g = \int_{\mathbb{R}} g(t)\,e^{iHt}Ae^{-iHt}\,\mathrm{d}t 
= \sum_\nu \widehat{g}(-\nu) A_\nu$, we have 
\begin{equation*}
   E_g^{(\Lambda_0)} := A_g - A_g^{(\Lambda_0)} 
= \int_{\mathbb{R}} \eta_{g,\Lambda_0}(t)\,e^{iHt}Ae^{-iHt}\,\mathrm{d}t,
\end{equation*}
and therefore $\|E_g^{(\Lambda_0)}\| \leq \|A\|\,\|\eta_{g,\Lambda_0}\|_{L^1} 
\leq \|\eta_{g,\Lambda_0}\|_{L^1}$.

We shall use the following basic estimate:
\begin{equation} \label{eq:basicl1}
    \|\eta\|_{L^1}
\leq \frac{1}{\pi}\bigl(\|m\|_{L^1} + \|m''\|_{L^1}\bigr),
\end{equation}
where $m \in W^{2,1}(\mathbb{R})$ and
$\eta(t) = \frac{1}{2\pi}\int m(\nu)\,e^{\pm i\nu t}\,\mathrm{d}\nu$. Indeed, 
for $|t| \leq 1$ we have $|\eta(t)| \leq \frac{1}{2\pi}\|m\|_{L^1}$, while 
for $|t| > 1$, integrating by parts twice gives 
$|\eta(t)| \leq \frac{1}{2\pi t^2}\|m''\|_{L^1}$. Integrating over $|t| \leq 1$ 
and $|t| > 1$ yields the bound.

We claim that for $j = 0, 1, 2$ and $\tau = \Omega(\beta)$, it holds that for some constant $c_g$, 
\begin{equation*}
    \widehat{g}^{(j)}(\nu) = \cO(\tau^j(1 + \tau|\nu|)^j e^{- c_g\tau^2\nu^2})\,, \q g \in \{f,\tilde{f}\}.     
\end{equation*}

\noindent\textbf{Centered Gaussian.} For $g = f$, we have 
\begin{equation*}
    \widehat{f}(\nu) = e^{-\tau^2\nu^2/2}, 
\end{equation*}
and for $j = 0, 1, 2$ there holds
\begin{equation*}
    |\widehat{f}^{(j)}(\nu)|
\leq \tau^j(1 + \tau|\nu|)^j e^{-\tau^2\nu^2/2}.    
\end{equation*}

\noindent\textbf{Shifted Gaussian.} For $g = \tilde{f}$, recall 
\begin{equation*}
    \widehat{\tilde{f}}(\nu)
= e^{\beta^2/(8\tau^2)}
\exp\!\left(-\frac{\tau^2}{2}\left(\nu - \mu\right)^2\right),
\qquad \mu := \frac{\beta}{2\tau^2}.
\end{equation*}
Using $(\nu - \mu)^2 \geq \frac{\nu^2}{2} - \mu^2$, we obtain
\begin{equation*}
    |\widehat{\tilde{f}}(\nu)|
\leq e^{\beta^2/(8\tau^2)} e^{\tau^2\mu^2/2} e^{-\tau^2\nu^2/4}
= e^{\beta^2/(4\tau^2)} e^{-\tau^2\nu^2/4}. 
\end{equation*}
Since $\tau = \Omega(\beta)$, the prefactor $e^{\beta^2/(4\tau^2)}$ is 
$\mathcal{O}(1)$. The derivatives satisfy
\begin{equation*}
    \widehat{\tilde{f}}'(\nu) = -\tau^2(\nu - \mu)\widehat{\tilde{f}}(\nu),
\qquad
\widehat{\tilde{f}}''(\nu) 
= \bigl(\tau^4(\nu-\mu)^2 - \tau^2\bigr)\widehat{\tilde{f}}(\nu), 
\end{equation*}
hence for $j = 0, 1, 2$ there holds
\begin{align*}
    |\widehat{\tilde{f}}^{(j)}(\nu)| = \cO(\tau^j(1 + \tau|\nu|)^j e^{-\tau^2\nu^2/4}).
\end{align*}

It remains to estimate $\|\eta_{g,\Lambda_0}\|_{L^1}$. Let
$\lambda:=\tau\Lambda_0$ and set
\[
    M_{g,\lambda}(x):=m_{g,\Lambda_0}(x/\tau).
\]
If $\zeta_{g,\lambda}(r):=(2\pi)^{-1}\int_{\mathbb R}
M_{g,\lambda}(x)e^{-ixr}\,\rd x$, then
$\eta_{g,\Lambda_0}(t)=\tau^{-1}\zeta_{g,\lambda}(t/\tau)$ and hence
\begin{equation} \label{auxeql1}
\|\eta_{g,\Lambda_0}\|_{L^1}=\|\zeta_{g,\lambda}\|_{L^1}
\end{equation}

For $\lambda\le 1$, we can show 
\begin{equation*}
\|\eta_{g,\Lambda_0}\|_{L^1}=\|\zeta_{g,\lambda}\|_{L^1} = \cO(1). 
\end{equation*}
Indeed, write
$G_g(x):=\widehat g(-x/\tau)$ and decompose
$M_{g,\lambda}=(1-\chi(x/\lambda))G_g=G_g-\chi(x/\lambda)G_g$. The inverse
Fourier transform of $G_g$ is
\[
    \frac{1}{2\pi}\int_{\mathbb R}\widehat g(-x/\tau)e^{-ixr}\,\rd x
    =\tau g(\tau r),
\]
and hence has $L^1$ norm $\|g\|_{L^1}=\cO(1)$ for $g\in\{f,\tilde f\}$ when
$\tau=\Omega(\beta)$. Next, let
\[
    \kappa(r):=\frac{1}{2\pi}\int_{\mathbb R}\chi(x)e^{-ixr}\,\rd x,
    \qquad
    \kappa_\lambda(r):=\frac{1}{2\pi}\int_{\mathbb R}\chi(x/\lambda)e^{-ixr}\,\rd x
    =\lambda\kappa(\lambda r).
\]
Since $\chi\in C_c^\infty(\mathbb R)$, $\kappa$ is Schwartz and
$\|\kappa_\lambda\|_{L^1}=\|\kappa\|_{L^1}<\infty$. Therefore the inverse
Fourier transform of $\chi(x/\lambda)G_g(x)$ is
$\kappa_\lambda*(\tau g(\tau\cdot))$, whose $L^1$ norm is bounded by
$\|\kappa\|_{L^1}\|g\|_{L^1}=\cO(1)$. Thus
$\|\zeta_{g,\lambda}\|_{L^1}=\cO(1)$ for $\lambda\le1$, finishing the proof of the claim.

Assume now that $\lambda\ge 1$. In the scaled variable $x=\tau\nu$, the
derivatives of $\widehat g(-x/\tau)$ up to order two are bounded by
\[
    \left|\frac{\rd^j}{\rd x^j}\widehat g(-x/\tau)\right|
    \le C(1+|x|)^j e^{-c_g x^2},
    \qquad j=0,1,2.
\]
Moreover, the derivatives of $\chi(x/\lambda)$ are supported in
$\{\lambda/2\le |x|\le \lambda\}$ and are bounded by
$\cO(\lambda^{-1})$ and $\cO(\lambda^{-2})$. Since $\lambda\ge1$, applying
\eqref{eq:basicl1} to $M_{g,\lambda}$ gives
\begin{equation*}
\|\eta_{g,\Lambda_0}\|_{L^1}=\|\zeta_{g,\lambda}\|_{L^1}
\le \frac{1}{\pi}\left(\|M_{g,\lambda}\|_{L^1}+\|M_{g,\lambda}''\|_{L^1}\right)
\le C(1+\lambda)^3 e^{-c'_g\lambda^2}. \qedhere
\end{equation*}
\end{proof}

\begin{lemma}[{Quasi-locality in energy of $\mc{T}'^\dagger[I]$}] 
\label{lem:quasi-local-verify}
Let $K_1$ and $K_2$ be defined as in \eqref{def:k1} and \eqref{def:k23}, respectively, and $O = \mc{T}'^\dagger[I]$ be given in \eqref{eq:operatoro}. Given  $\epsilon\in(0,1/4]$, let $\Lambda$ and $\epsilon'$ be as in \eqref{eq:paraomegapeo} required by Proposition \ref{thm:implementR}. 
Assume that the width parameter $\tau$ of $f$ satisfies  
\begin{equation*}
    \tau = \max\left\{\Omega(\beta),\Omega\left(\frac{k}{\Lambda}\,
\sqrt{\log\!\Bigl(\frac{k}{c^2\epsilon'}\Bigr)}\right)\right\}
\end{equation*}
with $k = \Theta(\log(1/\epsilon'))$, and $u$ satisfies $|u|\le 1/(2\|A_{\tilde f}\|)$ from Theorem \ref{thm:sqrt-be-lcu}. Then, $O$ is $(\Lambda,\epsilon')$-quasi-local in energy with respect to $H$.
Equivalently, there exists an operator $\widetilde O$ such that
\[
\widetilde O_\nu=0 \quad\text{for all } |\nu|\ge \Lambda,
\qquad
\|O-\widetilde O\|\le \epsilon'.
\]
\end{lemma}

\begin{proof}
We write 
\begin{equation*}
       O = c^2(I + uA_f) + c^2 X^\dagger X, \qquad X := \sqrt{I + uA_{\tilde{f}}}.
\end{equation*}
Since $\|A_{\tilde{f}}\| = \cO(1)$ for $\tau = \Omega(\beta)$, choosing
$|u| \leq 1/(2\|A_{\tilde{f}}\|)$ ensures $ r = |u|\,\|A_{\tilde{f}}\| \leq 1/2$.  Let $
P_k(z):=\sum_{j=0}^{k}\binom{1/2}{j}z^j$
be the degree-$k$ Taylor polynomial of $\sqrt{1+z}$ with  
\begin{equation*}
    \|X-P_k(u A_{\tilde f})\|
\le \frac{r^{k+1}}{1-r}.
\end{equation*}
With $k = \Theta(\log(1/\epsilon'))$, we can ensure 
\begin{equation} \label{auxeq:1}
    \|X^\dagger X-P_k(uA_{\tilde f})^\dagger P_k(uA_{\tilde f})\| = \cO(\epsilon'/c^2),
\end{equation}
since both $\|X\|$ and
$\|P_k(uA_{\tilde f})\|$ are $\cO(1)$. Therefore, if we define 
\[
O^{(k)} := c^2(I+uA_f) + c^2 P_k(uA_{\tilde f})^\dagger P_k(uA_{\tilde f}),
\]
then
\begin{equation}\label{auxeq00}
    \|O-O^{(k)}\| = \cO(\epsilon').     
\end{equation}

Next, let
\begin{equation*}
    \Lambda_0:=\frac{\Lambda}{2k},
\end{equation*}
and define the truncated operators $A_f^{(\Lambda_0)}$ and $A_{\tilde f}^{(\Lambda_0)}$ as in \eqref{def:oppag}. By Lemma \ref{lem:gaussian-bandlimit}, 
we choose $\tau = \Omega\left(\frac{k}{\Lambda}\sqrt{\log\!\left(\frac{k}{c^2\epsilon'}\right)}\right)$
large enough such that 
\begin{equation} \label{auxeq:2}
   \max_{g \in \{f,\tilde{f}\}} \|A_g - A_g^{(\Lambda_0)}\|
\le \min\!\Bigl\{\frac{1}{4|u|},\, \Theta\left(\frac{\epsilon'}{c^2}\right)\Bigr\}.
\end{equation}
Set $B := uA_{\tilde{f}}$ and $\widetilde{B} := uA_{\tilde{f}}^{(\Lambda_0)}$. 
Then $\|B\| \leq \tfrac{1}{2}$ and $\|\widetilde{B}\| \leq \tfrac{3}{4}$. 
The telescoping identity gives, for each $j \geq 1$,
\[
    \|B^j - \widetilde{B}^j\|
    \leq j\!\left(\tfrac{3}{4}\right)^{j-1}\!\|B - \widetilde{B}\|
    = j\!\left(\tfrac{3}{4}\right)^{j-1}|u|\,\|A_{\tilde{f}} - A_{\tilde{f}}^{(\Lambda_0)}\|.
\]
which implies, by $\sum_{j \geq 1} |\binom{1/2}{j}|\, j\, (\tfrac{3}{4})^{j-1} < \infty$ and \eqref{auxeq:2},
\[
    \|P_k(B) - P_k(\widetilde{B})\|
    \leq \sum_{j=1}^k \left|\binom{1/2}{j}\right| \|B^j - \widetilde{B}^j\| = \cO\left(\frac{\epsilon'}{c^2}\right).
\]
Noting $\|P_k(B)\|, \|P_k(\widetilde{B})\| = \cO(1)$, we have 
\begin{equation} \label{auxeq:3}
      \|P_k(B)^\dagger P_k(B) - P_k(\widetilde{B})^\dagger P_k(\widetilde{B})\|
   = \cO\left(\frac{\epsilon'}{c^2}\right).
\end{equation}

We now define
\begin{equation*}
    \widetilde{O} := c^2\bigl(I + uA_f^{(\Lambda_0)}\bigr)
    + c^2 P_k\!\bigl(uA_{\tilde{f}}^{(\Lambda_0)}\bigr)^\dagger
    P_k\!\bigl(uA_{\tilde{f}}^{(\Lambda_0)}\bigr).
\end{equation*}
and show that it is the desired approximation of $O$. Indeed, since $A_g^{(\Lambda_0)}$ has Bohr-frequency support in $[-\Lambda_0, \Lambda_0]$, 
the first term has support in $[-\Lambda_0, \Lambda_0]$ and 
$P_k(uA_{\tilde{f}}^{(\Lambda_0)})^\dagger P_k(uA_{\tilde{f}}^{(\Lambda_0)})$ has 
support in $[-2k\Lambda_0, 2k\Lambda_0] = [-\Lambda, \Lambda]$. Hence $\widetilde{O}$ 
has Bohr-frequency support in $[-\Lambda, \Lambda]$. For the approximation error, by \eqref{auxeq:2} and \eqref{auxeq:3}, we have 
\begin{align*}
    \| \widetilde{O} - O^{(k)}\|
    &\leq c^2|u|\,\|A_f - A_f^{(\Lambda_0)}\| + c^2\|P_k(B)^\dagger P_k(B) 
    - P_k(\widetilde{B})^\dagger P_k(\widetilde{B})\| = \cO(\epsilon'). 
\end{align*}
Combining it with \eqref{auxeq00}, we have $\|O-\widetilde O\| = \cO(\epsilon')$, and hence $O$ is $(\Lambda, \epsilon')$-quasi-local in energy with respect to $H$.
\end{proof}

We finally provide the proof of Theorem \ref{thm:rejection} for completeness.

\begin{proof}[Proof of Theorem \ref{thm:rejection}]
The trace-preserving completion and preservation of $\sigma_\beta$-KMS detailed
balance follow from the rejection-branch construction of~\cite{gilyen2026}.
It remains to verify the implementation hypotheses. Set the accuracy parameter
in Proposition~\ref{thm:implementR} to be $\epsilon_{\rm R}:=\Theta(\varepsilon')$,
and take $\alpha=\Theta(\|H\|)$. Since $O=\mathcal T'^\dagger[I]\le 3c^2I$ by
\eqref{eq:boundtio}, choosing
$c=\cO(1/\log(\beta\|H\|))$ with a sufficiently small implicit constant ensures
the small-norm hypothesis $\|O\|\le 10^{-3}/s^2$ required in
Proposition~\ref{thm:implementR}. Let
\begin{equation*}
    \Lambda=\frac{1}{10\bigl\lfloor \log_2(1/\epsilon_{\rm R})\bigr\rfloor},
\qquad
\eta_{\rm ql}= \cO\left(\frac{\epsilon_{\rm R}}{s^2\operatorname{polylog}(1/\epsilon_{\rm R})}\right)
\end{equation*}
be the quasi-locality parameters required in Proposition~\ref{thm:implementR}.
For the stated choice of $\tau$, with constants chosen large enough to dominate
the lower bound in Lemma~\ref{lem:quasi-local-verify}, and for sufficiently
small $u$, Lemma~\ref{lem:quasi-local-verify} gives that $O$ is
$(\Lambda,\eta_{\rm ql})$-quasi-local in energy with respect to $H$.
Therefore, all assumptions of Proposition~\ref{thm:implementR} are satisfied,
and we obtain an $\epsilon_{\rm R}$-approximate block encoding of $K/2$ with the
stated complexity. Equivalently, this is a block encoding of $K$ with
normalization $2$ and error $\cO(\epsilon_{\rm R})$; choosing the constant in
$\epsilon_{\rm R}=\Theta(\varepsilon')$ gives the target accuracy
$\varepsilon'$.

It remains to implement a block encoding of $O=K_1^\dagger K_1+K_2^\dagger K_2$.
By Lemma~\ref{lem:K1K2-implementation}, for any sufficiently small target
accuracy $\delta_{\rm K}$, both $K_1$ and $K_2$ admit
$\delta_{\rm K}$-approximate block encodings. Applying the standard
block-encoding product construction gives block encodings of
$K_1^\dagger K_1$ and $K_2^\dagger K_2$, and then the standard
linear-combination construction gives a block encoding of their sum $O$.
Taking
\begin{equation*}
    \delta_{\rm K}=\Theta\Bigl(\varepsilon'/\operatorname{polylog}(\beta\|H\|/\varepsilon')\Bigr)
\end{equation*}
so that the composition errors remain within $\cO(\varepsilon')$, we obtain a
block encoding of $O$ with the desired cost, by
Lemma~\ref{lem:K1K2-implementation} and Lemma~\ref{lem:block}.
\end{proof}

\end{document}